%% file: main.tex
\documentclass{sig-alternate-10}
\pdfoutput=1
\usepackage{graphicx}
\usepackage{enumerate}
\usepackage{xspace}
\usepackage{comment}
\usepackage{balance}  
\usepackage{microtype}
\usepackage{tikz}
\usepackage{times}
\usepackage{multirow}
\usepackage[normalem]{ulem}

\usetikzlibrary{arrows}

\newtheorem{definition}{Definition}[section]
\newtheorem{example}[definition]{Example}

\newtheorem{theorem}[definition]{Theorem}

\newtheorem{lemma}[definition]{Lemma}

\newcommand{\calA}{\mathcal{A}\xspace}
\newcommand{\calB}{\mathcal{B}\xspace}
\newcommand{\calF}{\mathcal{F}\xspace}

\newcommand{\famop}[2]{\calF^{#1}_{#2}\xspace}

\newcommand{\attA}{\mathsf{A}\xspace}
\newcommand{\attB}{\mathsf{B}\xspace}
\newcommand{\attK}{\mathsf{K}\xspace}

\newcommand{\lgeq}{\leq\geq}
\newcommand{\range}{\mathsf{R}}
\newcommand{\rangeunion}{\mathsf{U}}
\newcommand{\getsadd}{\gets+}
\newcommand{\affine}{\mathsf{aff}}

\newcommand{\Dbeg}{D_\mathsf{S}\xspace}
\newcommand{\Dend}{D_\mathsf{T}\xspace}
\newcommand{\Rbeg}{R_\mathsf{S}\xspace}
\newcommand{\Rend}{R_\mathsf{T}\xspace}

\newcommand{\NP}{\mathsf{NP}}
\newcommand{\diff}{\Delta}
\newcommand{\bd}{\mathrm{BD}}
\newcommand{\bdI}{\mathrm{BD}1}
\newcommand{\bdII}{\mathrm{BD}\mu}
\newcommand{\cost}{\mathit{cost}}
\newcommand{\ddiff}{\textsc{Data-Diff}\xspace}

\newenvironment{denselist}{
    \begin{list}{\small{$\bullet$}}%
    {\setlength{\itemsep}{1pt} \setlength{\topsep}{1pt}
    \setlength{\parsep}{1pt} \setlength{\itemindent}{0pt}
    \setlength{\leftmargin}{1.5em}
    \setlength{\partopsep}{1pt}}}%
    {\end{list}}

\newcommand{\topic}[1]{\vspace{-3.5pt}\smallskip \smallskip \noindent{\bf #1.}}
\newcommand{\stitle}[1]{\vspace{0.5em}\noindent\textbf{#1}}

\begin{document}


\title{Towards a Theory of {\LARGE  \sc Data-Diff}: Optimal Synthesis of \\ Succinct Data Modification Scripts}



%
%
%
%

\numberofauthors{3} 

\author{
%
%
\alignauthor
Tana Wattanawaroon\\
       \affaddr{University of Illinois (UIUC)}\\
       \email{wattana2@illinois.edu}
\alignauthor
Stephen Macke\\
       \affaddr{University of Illinois (UIUC)}\\
       \email{smacke@illinois.edu}
\alignauthor
Aditya Parameswaran\\
      \affaddr{University of Illinois (UIUC)}\\
      \email{adityagp@illinois.edu}
}

\maketitle

\begin{abstract}
\input{abstract.tex}
\end{abstract}

\input{intro.tex}
\input{related.tex}
\input{problem.tex}
\input{solution-ab-compose.tex}
\input{solution-aab-compose.tex}
\input{futurework.tex}

\bibliographystyle{abbrv}
\bibliography{main}  


\newpage
\begin{appendix}
\input{appendix.tex}
\end{appendix}

\end{document}

%% file: abstract.tex

This paper addresses the \ddiff problem: given
a dataset and a subsequent version of the dataset, find
the shortest sequence of operations that transforms the
dataset to the subsequent version, under a restricted
family of operations.
We consider operations similar to SQL UPDATE,
each with a \emph{condition} (WHERE)
that matches a subset of tuples
and a \emph{modifier} (SET)
that makes changes to those matched tuples.
We characterize the problem
based on different constraints on the attributes
and the allowed conditions and modifiers,
providing complexity classification
and algorithms in each case.

%% file: intro.tex
\section{Introduction}
\label{sec:intro}

Over the course of data analysis, data scientists
routinely generate versions of datasets by
performing various data curation and cleaning operations,
including updating, normalizing, fixing, adding, or
deleting attribute values or rows,
or adding or deleting new features or columns.
They may use various ad-hoc tools for
performing these edit operations,
including scripting tools like sed, awk, or perl,
or programming languages, like R or Python.
Each such new dataset version is stored in a networked file
system and shared
with other data scientists~\cite{bhardwaj2014datahub,maddox2016decibel,bhattacherjee2015principles}.
Usually, however, the sequence of edit operations or the script that
was used to generate the new version is not recorded along with the
new version---since it may have been the result of a quick-and-dirty update;
and even if the script is recorded, since
the script may be in various programming or scripting languages,
it may be hard to decipher or reverse-engineer the sequence of
edit operations performed within this script.

To tackle this issue, in this paper, we introduce
the \ddiff problem: {\em given a dataset $\Dbeg$
and a subsequent dataset $\Dend$ that was derived from $\Dbeg$,
can we synthesize the most succinct sequence of edit operations,
$\Delta$, that transforms $\Dbeg$ to $\Dend$?}
Our target is SQL edit operations that can be efficiently
executed in relational databases.
We call this problem the \ddiff problem
as the {\em data}-analog of the traditional text {\em diff},
or differentiation problem,
often used in source code versioning systems to synthesize
the sequence of edit operations that resulted in a new version.

There are three reasons why solving \ddiff,
i.e., synthesizing a succinct
sequence of edit operations, is valuable: {\em understanding, generalization,
and compactness}.
First, the data-diff helps users compactly understand the edit operations
that have been made to generate a new version $\Dend$ from $\Dbeg$, without having
to read through a long programming script;
second, it allows us to potentially record and recreate the edit
operations so that they can be similarly applied to other datasets; and
third, instead of storing $\Dend$, we can simply store the sequence of edit
operations, which, since it is written in SQL, is often smaller.

\topic{Our Focus}
In this paper, {\em our key contribution is to
introduce the \ddiff problem and study it from a theoretical perspective,
aiming to characterize the complexity of the problem and understand
when the problem becomes intractable}.
We focus on recovering edits to a single relation $R$,
with edit operations that follow the following template:
\[\text{$\mathtt{UPDATE}$ $R$ $\mathtt{SET}$ $\langle U \rangle$ $\mathtt{WHERE}$ $\langle C \rangle$;}\]
We characterize the complexity of \ddiff across three dimensions:
\begin{enumerate}
\itemsep0em
\item {\em [characteristics]} the attributes that may be used within $U$ and $C$: we call an
attribute {\em read-only}  ({\em write-only})  if it can be used within $C$ ($U$) but not
$U$ ($C$), {\em read-write}  if it can be used within $U$ and $C$,
and {\em inaccessible}  if should not be used within either $U$ or $C$;
\item {\em [modifiers]} the space of transformations that can be used within $U$: we span
basic assignment operations, as well as arithmetic operations; and
\item {\em [conditions]} the space of conditions that can be used within $C$: we span
both equality conditions, $\leq$ and $\geq$, and range-based conditions.
\end{enumerate}
In any of these cases, the user will specify the space
of attribute characteristics, modifiers, and conditions,
and the system will then automatically synthesize the smallest
sequence of edit operations.
Next, we illustrate the challenges in solving \ddiff using a simple example.

\input{ab-compose-table.tex}

\begin{figure}[t]
\centering
\begin{tabular}{c c c c c}
$\attK$&$\attA$&$\attB$\\
\hline
\texttt{c17}&$1$&$0$\\
\texttt{3bd}&$5$&$0$\\
\texttt{97a}&$3$&$0$\\
\texttt{1b8}&$0$&$0$\\
\texttt{94f}&$4$&$0$\\
\texttt{842}&$2$&$0$\\
\\
&$R_1$
\end{tabular}
\begin{tikzpicture}
\path[draw=black,solid,line width=1mm,fill=black,
    preaction={-triangle 90,thin,draw,shorten >=-1mm}
] (0, -0.3) -- (0.3, -0.3);
\end{tikzpicture}
\begin{tabular}{c c c c c}
$\attK$&$\attA$&$\attB$\\
\hline
\texttt{c17}&$1$&$1$\\
\texttt{3bd}&$5$&$3$\\
\texttt{97a}&$3$&$2$\\
\texttt{1b8}&$0$&$1$\\
\texttt{94f}&$4$&$2$\\
\texttt{842}&$2$&$1$\\
\\
&$R_2$
\end{tabular}
\begin{tikzpicture}
\path[draw=black,solid,line width=1mm,fill=black,
    preaction={-triangle 90,thin,draw,shorten >=-1mm}
] (0, -0.3) -- (0.3, -0.3);
\end{tikzpicture}
\begin{tabular}{c c c c c}
$\attK$&$\attA$&$\attB$\\
\hline
\texttt{c17}&$7$&$1$\\
\texttt{3bd}&$8$&$3$\\
\texttt{97a}&$8$&$2$\\
\texttt{1b8}&$7$&$1$\\
\texttt{94f}&$8$&$2$\\
\texttt{842}&$7$&$1$\\
\\
&$R_3$
\end{tabular}
\caption{Example of a setting with three versions of a relation $R$
where we want to solve \ddiff with one read-write column $\attA$, and one
write-only column $\attB$.}
\label{fig:ddiffex}
\end{figure}

\begin{example}[Motivating Example]\label{ex:motivating-example}
Consider the scenario  in Figure~\ref{fig:ddiffex}, where
we depict three versions of a given relation $R$,
namely $R_1$, $R_2$, and $R_3$,
with the primary key $\attK$.
Using $\attK$, we can identify how individual tuples have evolved
across the versions.
For this simple example, we do not have any tuples being added or deleted,
nor do we have any attributes being added or deleted.
Our goal is to solve \ddiff under the specification
that we have
one read-write attribute, $\attA$, and one write-only attribute, $\attB$.
(since $\attK$ is the primary key, in this case, it has
been denoted an inaccessible attribute, which means that
it cannot be used in the modifier or in the condition.)

One approach to solving \ddiff between $R_1$ and $R_2$,
which only differ in the value
of $\attB$, is to use six edit operations of the following form:
\[\text{$\mathtt{UPDATE}$ $R$ $\mathtt{SET}$ $\attB=b_k$ $\mathtt{WHERE}$ $\attA=a_k$;}\]
one for each tuple. Recall that $\attA$ being a read-write attribute,
can be used for the (equality) condition, while $\attB$ being a write-only attribute
can be used for the (assignment) modifier.
If we relax the space of conditions to admit $\leq$ and $\geq$
in addition to equality,
then there is a shorter sequence of three edit operations:
\begin{align*}
& U_1: \text{$\mathtt{UPDATE}$ $R$ $\mathtt{SET}$ $\attB=1$ $\mathtt{WHERE}$ $\attA\leq 2$;}\\
& U_2: \text{$\mathtt{UPDATE}$ $R$ $\mathtt{SET}$ $\attB=2$ $\mathtt{WHERE}$ $\attA\geq 3$;}\\
& U_3: \text{$\mathtt{UPDATE}$ $R$ $\mathtt{SET}$ $\attB=3$ $\mathtt{WHERE}$ $\attA = 5$;}
\end{align*}
Notice that the order of operations is important: $U_1 \rightarrow U_2
\rightarrow U_3$ does not give the same result as $U_1 \rightarrow U_3
\rightarrow U_2$.
Similarly, to solve \ddiff between $R_2$ and $R_3$
(wherein the read-write attribute $\attA$ is transformed),
we could use as many as six operations, but
in fact two operations suffice:
\begin{align*}
& U_4: \text{$\mathtt{UPDATE}$ $R$ $\mathtt{SET}$ $\attA=7$ $\mathtt{WHERE}$ $\attA\leq 2$;}\\
& U_5: \text{$\mathtt{UPDATE}$ $R$ $\mathtt{SET}$ $\attA=8$ $\mathtt{WHERE}$ $\attA\leq 5$;}
\end{align*}
Once again, $U_5 \rightarrow U_4$ does not provide the same
result as $U_4 \rightarrow U_5$.
As it turns out, this sequence of three edit operations
for the first case, and two edit operations for the second case
are the smallest possible sequences, based on
modifiers that are assignment-based and
on conditions that are based on $\leq$, $\geq$ or equality.
Indeed, when we expand the space of modifiers to not just assignment, but also
addition or subtraction, the \ddiff problem becomes even more challenging.
Overall, depending on the instance, the smallest
sequence of operations may be as small as one operation, or
as many as $O(N)$ (typically non-commutative) operations, where $N$ is the number of tuples, making
it challenging to navigate.
\end{example}

\input{aab-compose-table.tex}

\topic{Related Work}
The \ddiff problem is related to the
{\em view synthesis} problem, a complementary
problem that targets the following setting:
given $D, D'$, find the most succinct
single view definition $Q$ using selection operations such that
$D' \approx Q(D)$~\cite{das2010synthesizing,QBO}.
For example,
\[Q(R): \text{$\mathtt{SELECT}$ $*$ $\mathtt{FROM}$ $R$ $\mathtt{WHERE}$ $\attA=k$;}\]
is a view definition $Q$ that selects all of the
tuples that match a certain criteria from $R$.
This work has been extended in multiple directions that we will discuss
in Section~\ref{sec:related}.
\ddiff is much harder than view synthesis,
due to non-commutativity of edit operations,
leading to intractability even for relations with
a finite number of attributes, while view synthesis
is only intractable when the number of attributes is allowed to vary.
\ddiff is also related to the problem
of synthesizing string transformations---the difference
between that line of work and ours is the difference
between learning regular expressions and learning SQL modification
statements: the space of operations and therefore the techniques
and contributions are very different.
We will cover related work in more detail in Section~\ref{sec:related}.

\topic{Contributions}
We introduce the family of \ddiff problems
under different attribute characteristics and the space
of modifiers and conditions of interest. We identify
a ``base case'', fully characterize it, and then
identify a generalization and proceed to show hardness
results in the generalization.
The characterization summary can be seen in Tables~\ref{table:summary1}
and~\ref{table:summary2}.

%% file: ab-compose-table.tex
\begin{table*}[!hbt]
\centering
\begin{tabular}{c l l l l}
& $\gets$ & $+$ & $\getsadd$ & $\affine$ \\
\hline\hline
$=$          & [Thm \ref{eqgets}] $O(N \log N)$
             & [Thm \ref{eqadd}] $O(N \log N)$
             & [Thm \ref{eqgetsadd}] $O(N \log N)$
             & [Thm \ref{eqaffine}] $O(N \log N)$\\
\hline
$\leq$       & [Thm \ref{leqgets}] $O(N \log N)$
             & [Thm \ref{leqadd}] $O(N \log N)$
             & [Thm \ref{leqgetsadd}] $O(N \log N)$
             & [Thm \ref{leqaffine}] $O(N \log N)$\\
\hline
\multirow{2}{*}{$\lgeq$} & \multirow{2}{*}{[Thm \ref{lgeqgets}] $O(N^2)$}
                         & [Thm \ref{lgeqadd}] $\NP$-hard
                         & \multirow{2}{*}{[Thm \ref{lgeqgetsadd}] $\NP$-hard}
                         & \multirow{2}{*}{[Thm \ref{lgeqaffine}] $\NP$-hard}\\
                         &
                         & [Thm \ref{lgeqaddapprox}] $O(N\log N)$ to $+1$-approx\\
\hline
\multirow{2}{*}{$\range$} & \multirow{2}{*}{[Thm \ref{rangegets}] $O(N^4)$}
                          & [Thm \ref{rangeadd}] $\NP$-hard
                          & \multirow{2}{*}{[Thm \ref{rangegetsadd}] $\NP$-hard}
                          & \multirow{2}{*}{[Thm \ref{rangeaffine}] $\NP$-hard}\\
                          &
                          & [Thm \ref{rangeaddapprox}] $O(N\log N)$ to $\times 2$-approx\\
\hline
\multirow{2}{*}{$\rangeunion$} & \multirow{2}{*}{[Thm \ref{uniongets}] $\NP$-hard}
                          & \multirow{2}{*}{[Thm \ref{unionadd}] $\NP$-hard}
                          & \multirow{2}{*}{[Thm \ref{uniongetsadd}] $\NP$-hard}
                          & \multirow{2}{*}{[Thm \ref{unionaffine}] $\NP$-hard}\\
\end{tabular}
\caption{Results summary for the one read-only attribute, one write-only attribute case}
\label{table:summary1}
\end{table*}

%% file: aab-compose-table.tex
\begin{table*}[!hbt]
\centering
\begin{tabular}{c l l l l}
& $\gets$ & $+$ & $\getsadd$ & $\affine$ \\
\hline\hline
$=$          & [Thm \ref{multeqgets}] $\NP$-hard
             & [Thm \ref{multeqrest}] $\NP$-hard
             & [Thm \ref{multeqrest}] $\NP$-hard
             & [Thm \ref{multeqrest}] $\NP$-hard\\
\hline
$\leq$       & ?
             & [Thm \ref{multleq}] $\NP$-hard
             & [Thm \ref{multleq}] $\NP$-hard
             & [Thm \ref{multleq}] $\NP$-hard\\
\hline
$\lgeq$       & ?
              & [Thm \ref{multtrivial}] $\NP$-hard
              & [Thm \ref{multtrivial}] $\NP$-hard
              & [Thm \ref{multtrivial}] $\NP$-hard\\
\hline
$\range$      & [Thm \ref{multrangegets}] $\NP$-hard
              & [Thm \ref{multtrivial}] $\NP$-hard
              & [Thm \ref{multtrivial}] $\NP$-hard
              & [Thm \ref{multtrivial}] $\NP$-hard\\
\hline
$\rangeunion$ & [Thm \ref{multuniongets}] $\NP$-hard
              & [Thm \ref{multtrivial}] $\NP$-hard
              & [Thm \ref{multtrivial}] $\NP$-hard
              & [Thm \ref{multtrivial}] $\NP$-hard\\
\end{tabular}
\caption{Results summary for the multiple read-only attributes, one write-only attribute case}
\label{table:summary2}
\end{table*}

%% file: related.tex
\section{Related Work}
\label{sec:related}

\ddiff is related to the topics of view synthesis
and learning string transformations from examples.

\topic{View Synthesis} The view synthesis problem originally defined the question
of synthesizing a view definition given two database instances, which was originally laid out in Das Sarma et al.~\cite{das2010synthesizing} and Tran et al.~\cite{QBO} and
extended in various ways since
then~\cite{lu2013quicksilver,zhang2013automatically,tran2014query,zhang2013reverse,panev2016reverse}.
For example, recent work has extended the original work
on the view synthesis problem to the problem of
synthesizing join queries~\cite{zhang2013reverse}
and top-$k$ queries~\cite{panev2016reverse}.
Other work has extended the view synthesis problem
to an iterative one, with the user being asked to confirm
the presence or absence of tuples one at a time in order to {\em learn} an appropriate user query  for various
settings~\cite{bonifati2014paradigm,bonifati2016learning,bonifati2015learning,bonifati2014interactive,abouzied2013learning}.
Earlier work studied the problem of checking if there
exists a view definition without synthesizing it~\cite{fletcher2009expressive}.
Another related direction is that of synthesizing a view given
multiple pairs of database instances, introduced in the context of
data integration as a problem of learning schema mappings from
data examples~\cite{cate2013learning,gottlob2010schema,alexe2011designing}.

While all of these directions are interesting and relevant to
the \ddiff problem, note that the \ddiff problem
is substantially harder than the view synthesis problem, even
when applied on a single relation $R$.
First, edit operations, unlike selection operations,
are non-commutative and therefore cannot be applied in any order.
Thus, the order of operations, while unimportant in view synthesis,
is crucial in \ddiff.
Second, the ability to use multiple operations is not very
important in the view synthesis problem,
since we can simply overload the $\mathtt{WHERE}$ clause
to be more complex; in the \ddiff problem on the other hand,
multiple edit operations offer substantial additional power, e.g.,
transforming $R_1$ to $R_3$ as given in Figure~\ref{fig:ddiffex}
would be difficult using one operation.

For these reasons, we find that the problem of \ddiff becomes
intractable much sooner---even on edit operations on a single relation with
two or three attributes, while the view synthesis problem
is only intractable when the number of attributes is allowed to vary.
In fact, notice that \ddiff problem has a view synthesis problem
as a sub-problem: for the case where a number of tuples have been
deleted from $D$ to $D'$, we could use the results from the view
synthesis problem to identify the condition that selects
all of the tuples to be deleted, and therefore
we can inherit all of the same hardness results for those cases.
To understand the complexity of \ddiff independent of view synthesis,
we focus on the case when no tuples have been deleted.

\topic{String Transformations}
A related direction from the program analysis community focuses
on the learning of string transformations given input-output examples~\cite{DBLP:conf/popl/Gulwani11,DBLP:journals/cacm/GulwaniHS12},
extending it to various settings in cleaning data in spreadsheets, such as
transforming times and dates~\cite{DBLP:journals/pvldb/SinghG12},
numbers~\cite{singh2012synthesizing},
text~\cite{kini2015flashnormalize}, and
miscellaneous data types~\cite{singh2016transforming},
changing the structure of spreadsheet tables~\cite{harris2011spreadsheet}, as well as
extracting structured data from semi-structured spreadsheet data~\cite{barowy2015flashrelate}.
Like us, this body of work targets {\em edit operations}---however, these operations are
regular-expression like operations that are applied
to transform each value in a set of values (e.g., extracting the first three digits of a phone number).
Each such value can be then treated as a
training example for learning the edit operation.
Instead, we focus SQL operations:
not as fine-grained at the value level,
but are more fine-grained at a global level, admitting conditional clauses, e.g.: if $\attA \in [a, b]$, add $c$ to $\attB$.
Thus, the difference in the space of operations under consideration can be
seen as the difference between regular expressions being applied to a set of values (in the string transformation case),
versus a sequence of SQL modification statements (in our case).
In addition, we do not attempt to precisely characterize
the complexity of learning transformations as a function of the space of operations, preferring
instead to prove soundness and completeness.

%% file: problem.tex
\section{Problem Definitions}
\label{sec:problem}

In this section, we formulate the problem of finding a
succinct description of changes between two datasets.
We define the \emph{diff}, which captures the notion of
the description of changes, along with some relevant terms.
Then we formally define the problem and scope
of operations that are of interest in this paper.

To understand and characterize the complexity frontier
of the \ddiff problem, where the goal is to find the most
succinct sequence of operations that transform $\Dbeg$ to $\Dend$,
we assume that $\Dbeg$ and $\Dend$ are both single relations $\Rbeg$ and $\Rend$
with the same schema, along with an unmodified primary key attribute
(e.g., $\mathtt{employeeID}$, $\mathtt{transactionID}$)
that allows us to track how tuples have evolved---thus,
there is a one-to-one correspondence between the tuples in $\Rbeg$ and $\Rend$.
We further assume that the primary key values in $\Rbeg$ and $\Rend$ are the same,
essentially guaranteeing that there are no insertions or deletions.
Thus, overall, our setting is one where there is a single
relation (with a primary key)
being modified by {\em data modification operations},
but there are no insertions or deletions (of tuples or attributes),
or modification of schema.
We will formalize these assumptions later in this section.

\stitle{Rationale for Assumptions.}
We now briefly describe why we make these simplifying assumptions
to focus on \ddiff for {\em data modification operations}.
When there is an unmodified primary key,
insertions of new tuples are easy to identify,
and trivial to represent as either a single batch $\mathtt{INSERT}$ statement,
or insertion of one tuple at a time,
with no further compression possible or necessary.
Deletions of tuples, on the other hand, ends up being
equivalent to the view synthesis problem (as described in Section~\ref{sec:related}), since we need to identify a query $Q$ that selects precisely
the tuples that were deleted, and thus we can reuse existing
results from related work previously discussed.
Since there is an unmodified primary key, if we know which attributes
are deleted, they can all be dropped in one single $\mathtt{ALTER}$
statement, along with any attributes that are renamed.
Naturally, attributes that are inserted are a lot more complicated,
since, in general, a succinct description for new attributes
would fall under the realm of pattern recognition---this is outside
the scope of our work, which focuses on data modification.


\subsection{Similar Relations, Diff, and Best Diff}
First, we introduce the notion of attribute characteristics.
Different settings of characteristics play a major role in
determining the hardness of the problem.
Here, $\calA$ is the set of \emph{read} attributes on which conditions are based,
and $\calB$ is the set of \emph{write} attributes on which modifiers make changes.

In general, we can detect which attributes $\calB$
have been modified automatically, but we allow the user to specify the
set of attributes $\calA$ explicitly, since they may not
want the system to use all attributes to infer SQL data modification scripts.
For example, if the user knows that $\mathtt{gender}$ is never
an attribute that is read when modifying the $\mathtt{GPA}$, they can
exclude $\mathtt{gender}$ from the set of attributes in $\calA$.

\begin{definition}[Attribute Characteristics]
An attribute $\attA\in\calA\cup\calB$ is called
\emph{read-only} if $\attA\in\calA$ and $\attA\not\in\calB$, or
\emph{write-only} if $\attA\not\in\calA$ and $\attA\in\calB$, or
\emph{read-write} otherwise.
\end{definition}

Second, we define ``similar'' relations.
The \ddiff problem concerns two relations, one representing the ``before'' snapshot
and the other representing the ``after'' snapshot.
As previously discussed,
we will not consider adding or removing attributes, and we
want to exclude insertions and deletions of tuples from our family
of possible operations; we only consider ``update'' operations.
Thus, the two relations should have the same schema
and the same number of rows.

In addition, we want to be able to tell
which tuples map to which in the two relations,
hence the requirement that the two relations share a primary key $\attK$, and that
the sets of primary keys are identical and cannot be modified.
The primary key serves as an identifier of the tuples
in the two relations.

\begin{definition}[Similar Relation]
For an attribute $\attK$ and sets of attributes
$\calA$ and $\calB$, neither of which contains $\attK$,
two relations $\Rbeg$ and $\Rend$ are \emph{$(\attK,\calA,\calB)$-similar}
iff
\begin{itemize}
\itemsep0em
\item $\Rbeg$ and $\Rend$ both have schema $\{\attK\}\cup\calA\cup\calB$,
and $\attK$ is their primary key, and
\item $\pi_\attK(\Rbeg)=\pi_\attK(\Rend)$
(here $\pi$ is the projection operator in relational algebra).
\end{itemize}
\end{definition}
In other words, $\Rbeg$ and $\Rend$ have the same schema with
one primary key attribute containing the same set of values.
This implies that the number of tuples in $\Rbeg$ and $\Rend$ are equal,
and that we can match the tuples in $\Rbeg$ and $\Rend$ one-to-one based on
the primary key. Note that we could simply define two relations as
similar if they have the same schema, but our definition explicitly references
the sets $\calA$ and $\calB$ as a notational convenience that will help
with later exposition.

Third, we define what an operation is, and what it does.
It must obey the
read-write characteristics of the attributes.

\begin{definition}[Operation]
For sets of attributes $\calA$ and $\calB$,
an \emph{$(\calA, \calB)$-operation} $f=(p, u)$ has
a \emph{condition} $p$ on attributes in $\calA$ and
a \emph{modifier} $u$ on attributes in $\calB$. Let
$f(\Rbeg) = \Rend$ if and only if $\Rend$ is the resulting relation after
calling the SQL command:
\[\text{$\mathtt{UPDATE}$ $\Rbeg$ $\mathtt{SET}$ $u$ $\mathtt{WHERE}$ $p$.}\]
\end{definition}
Note that the result of an $(\calA, \calB)$-operation
is $(\attK,\calA,\calB)$-similar to the operand; i.e.,
if $\Rbeg$ and $\Rend$ are relations such that $f(\Rbeg)=\Rend$
and $\Rbeg$ has schema $\{\attK\}\cup\calA\cup\calB$,
then $\Rbeg$ and $\Rend$ are $(\attK,\calA,\calB)$-similar.

We now define the \emph{diff}, the sequence of operations transforming
the ``before'' relation to the ``after'' relation, along with its
associated cost.
In the following definitions, we consider an attribute $\attK$,
sets of attributes $\calA$ and $\calB$ neither of which contain $\attK$,
a set of $(\calA, \calB)$-operations $\calF$,
and $(\attK,\calA,\calB)$-similar relations $\Rbeg$ and $\Rend$.

\begin{definition}[Diff]
A sequence of operations $F=(f_1,\ldots,f_m)$
where $f_i\in\calF$ for each $i\in[m]$
is called a \emph{diff} between $\Rbeg$ and $\Rend$
under $\calF$,
also written $F(\Rbeg)=\Rend$,
if there are relations $R_0,\ldots,R_m$ such that
\begin{denselist}
    \item $R_0=\Rbeg$,
    \item $R_m=\Rend$, and
    \item $f_i(R_{i-1})=R_i$ for all $i\in[m]$.
\end{denselist}
Let $\diff(\Rbeg,\Rend,\calF)$
denote the set of all diffs between $\Rbeg$ and $\Rend$
under $\calF$.
\end{definition}

\begin{definition}[Cost]
Each operation $f$ has an associated integer \emph{cost}, denoted $\cost(f)$.
The cost of a diff $F=(f_1,\ldots,f_m)$ is defined as
$\cost(F)=\sum_{i=1}^m \cost(f_i)$.
\end{definition}

\begin{definition}[Best Diff]
A diff $F\in\diff(\Rbeg,\Rend,\calF)$
is called a \emph{best diff} between $\Rbeg$ and $\Rend$
under $\calF$
if it has the smallest cost in $\diff(\Rbeg,\Rend,\calF)$;
i.e., for any diff $F'\in\diff(\Rbeg,\Rend,\calF)$,
we have $\cost(F)\leq\cost(F')$.
We also write that $F$ is a best diff in
$\diff(\Rbeg,\Rend,\calF)$.
\end{definition}

Note that if $\diff(\Rbeg,\Rend,\calF)$
is nonempty, then it must contain a best diff,
by the well-ordering principle of integers.

\subsection{Diff Problems}
Next, we define the best diff problem
that is the focal point of this paper.

\begin{definition}[Best Diff Problem]
Fix a family of $(\calA, \calB)$-operations $\calF$.
The best diff problem $\bd(\calF)$ is, given as input:
\begin{denselist}
\item an attribute $\attK$,
\item attribute sets $\calA$ and $\calB$, neither of which contain $\attK$, and
\item two $(\attK,\calA,\calB)$-similar $N$-tuple relations $\Rbeg$ and $\Rend$,
where all values are integers,
\end{denselist}
find and return a best diff between $\Rbeg$ and $\Rend$
under $\calF$ if one exists,
or correctly report that no diffs exist.
\end{definition}

Note that while we
restrict relations $\Rbeg$ and $\Rend$ to integer values
(for simple arguments of representation sizes), conditions
and modifiers are not restricted to integers; real values
can be used.

The following auxiliary definitions are used in proofs.
\begin{definition}[Attribute Values]
For an attribute $\attA$, $V_\attA(\Rbeg,\Rend)$ is the set of all
$\attA$ values in relations $\Rbeg$ and $\Rend$; in other words,
\[V_\attA(\Rbeg,\Rend) = \pi_\attA(\Rbeg)\cup\pi_\attA(\Rend)\]
\end{definition}

\begin{definition}[Boundary and Length]
    Let
    \begin{align*}
        v_\attA^\mathrm{max} &= \max V_\attA(\Rbeg,\Rend)+1\\
        v_\attA^\mathrm{min} &= \min V_\attA(\Rbeg,\Rend)-1
    \end{align*}

    For an operation $f$, define the \emph{length} $\ell(f)$ as
    \[\ell(f)=\begin{cases}
    a-v_\attA^\mathrm{min}&\text{if $f$ has the condition $\attA\leq a$}\\
    v_\attA^\mathrm{max}-a&\text{if $f$ has the condition $\attA\geq a$}\\
    z-a+1&\text{if $f$ has the condition $\attA\in [a,z]$}\\
    \end{cases}\]
    For a sequence of operations
    $F=(f_1,\ldots,f_m)$, define the \emph{total length} of $F$
    as $\ell(F)=\sum_{i\in[m]}\ell(f_i)$.
\end{definition}


\subsection{Families of Operations}
Generally, $(\calA,\calB)$-operations can be simple or complicated.
Given two relations $\Rbeg$ and $\Rend$, one might claim that there is a diff
between them containing the following $(\calA,\calB)$-operation as its
only operation:
\[\left(\attA\in\{2,9,11,23\},\attB\gets\left\lceil|\attB|^{\sqrt{11}}\right\rceil-7\attB^2\right)\]
The given $(\calA,\calB)$-operation has an overfitting condition
and a complicated modifier, which makes it unlikely to be an operation
actually used to transform $\Rbeg$ into $\Rend$ by, say, an accountant
working on this database. Therefore, we would like to limit ourselves
to operations that are relatively simple and are more likely to
correspond to actual scenarios.

We describe families of $(\calA,\calB)$-operations
that are of interest in this paper.
Here, $\attA$ is an attribute from $\calA$,
and $\attB$ is an attribute from $\calB$.

\topic{Condition Types}
We consider conditions $p$ that are conjunctions of
single-attribute clauses,
i.e., statements in the form $p=p_1\wedge \ldots \wedge p_h$,
where the clauses $p_i$ have the same type but are on different attributes.
The condition $p$ on $\calA$ does not necessarily use
all attributes in $\calA$, but must use at least one (cannot be empty).

We consider the following single-attribute clause types.
\begin{center}
\begin{tabular}{c l l c}
    symbol & name & condition & cost\\
\hline
$=$         & equality & $\attA=a$ & $1$\\
$\leq$      & at-most & $\attA\leq a$ & $1$\\
$\lgeq$     & at-most/at-least & $\attA\leq a$ \textbf{or} $\attA\geq a$ & $1$\\
$\range$    & range & $\attA\in [a,z]$ & 1\\
$\rangeunion$ & union-of-ranges & $\attA\in\bigcup_{j=1}^r [a_j,z_j]$ & varies
\end{tabular}
\end{center}
The cost is $1$ \emph{per operation} (not per clause),
except in the union-of-ranges case, where
the cost is $\kappa_0+\kappa_1\sum r$, where $\sum r$
is the sum of number of ranges over all clauses.
Here, $\kappa_0$ and $\kappa_1$ are
non-negative integers to be supplied as input.

For the at-most/at-least clause type, each clause can assume
either of the two subtypes, and it is not required that all clauses
use the same subtype. The clause type with only \emph{at-least}
condition is not explicitly discussed, because it is symmetric to
using the at-most clause type.

\topic{Modifier Types}
We only consider single-attribute modifiers in this paper.
We consider the following modifier types.
\begin{center}
\begin{tabular}{c l l}
    symbol & name & modifier\\
\hline
$\gets$      & assignment & $\attB\gets b$\\
$+$          & increment  & $\attB\gets \attB+b$\\
$\getsadd$   & assignment/increment & $\attB\gets b$ \textbf{or} $\attB\gets \attB+b$\\
$\affine$    & affine & $\attB\gets b\attB + c$
\end{tabular}
\end{center}
The modifier type does not affect the cost of an operation.

\topic{Operations}
The family of $(\calA,\calB)$-operations using condition type $\phi$ and
modifier type $\omega$ is denoted by $\famop{\phi}{\omega}$.
For example, $\famop{\leq}{+}$ is the family of
$(\calA,\calB)$-operations where each operation uses an
\emph{at-most} condition and an \emph{increment} modifier.

\begin{example}
Once again, consider the three versions of the relation
$R$ given in Figure~\ref{fig:ddiffex},
namely $R_1$, $R_2$, and $R_3$.
Let $\calA=\{\attA\}$, $\calB=\{\attA,\attB\}$, so that $\attA$
is a read-write attribute and $\attB$ is a write-only attribute.

Let $f_1$ and $f_2$ be the following operations:
\begin{align*}
f_1&=(\attA\leq 2, \attA\gets 7)\\
f_2&=(\attA\leq 5, \attA\gets 8)
\end{align*}
Here, $f_1$ is in $\famop{\lgeq}{\gets}$
(and also $\famop{\leq}{\gets}$ and
$\famop{\lgeq}{\getsadd}$),
and so is $f_2$.

If $F=(f_1,f_2)$, then $F$ is a diff between $R_2$ and $R_3$
under $\famop{\lgeq}{\gets}$, and $\cost(F)=2$.
Note that $F'=(f_2,f_1)$ is, however, not a diff between $R_2$ and $R_3$
under $\famop{\lgeq}{\gets}$.
\end{example}

%% file: solution-ab-compose.tex
\section{Base Case: $\bdI$ problems}
\label{sec:solution-ab-compose}

In this section, we consider a ``base case'' of the best diff
problems in terms of number of attributes and attribute characteristics,
and present its characterization under different families of operations.

The $\bdI(\calF)$ problem is similar to the best diff $\bd(\calF)$
problem, but constrained to one read-only attribute,
one write-only attribute, and no read-write attributes.
Let $\calA=\{\attA\}$ and $\calB=\{\attB\}$,
where $\attA$ and $\attB$ are different attributes.

We also assume that
for any tuple $T_1\in \Rbeg$
and $T_2\in \Rend$, if $T_1.\attK=T_2.\attK$ then $T_1.\attA=T_2.\attA$,
because an $(\calA,\calB)$-operation cannot modify $\attA$ values.
If the assumption does not hold, we can immediately claim that a
diff between $\Rbeg$ and $\Rend$ does not exist.

Table \ref{table:summary1} summarizes the characterization.
The table is roughly ordered according to how ``powerful''
each condition/modifier type is, although it is not necessarily
true that a condition/modifier is a generalization of what precedes it.
We encounter the hardness boundary at the families of operations
$\famop{\lgeq}{+}$ and $\famop{\rangeunion}{\gets}$, where we present
two main $\NP$-hardness results via reductions from different problems.
While the remaining $\NP$-hardness results do not trivially follow
from the two main results, they use similar reductions.
Polynomial-time results are discussed more thoroughly
in Appendix~\ref{sec:polytimeproofs}.

\subsection{With Equality Conditions}
With \emph{equality} ($\attA=a$) conditions, tuples with different $\attA$ values
are independent of each other, in terms of how the $(\calA,\calB)$-operations
affect them. Therefore, the best diff problem under these families of
operations is rather straightforward.
\begin{theorem}
    The
    $\bdI(\famop{=}{\gets})$,
    $\bdI(\famop{=}{+})$,
    $\bdI(\famop{=}{\getsadd})$, and
    $\bdI(\famop{=}{\affine})$ problems
    can be solved in $O(N\log N)$ time.
\label{eqgets}
\label{eqadd}
\label{eqgetsadd}
\label{eqaffine}
\end{theorem}

\subsection{With At-most Conditions}
With \emph{at-most} ($\attA\leq a$) conditions, we can always reorder
the operations within a diff (with some modifications) so that
they affect the tuples in a certain order. Such reordering
allows for polynomial time algorithms under all families
of operations of interest.

More precisely, the at-most condition and the modifiers
permit the theorems to utilize this property:
if there is a best diff, there must
be a best diff $F=(f_1,\ldots,f_m)$
in which for all $i,j\in[m]$, if $i<j$ and
$f_i$ has condition $\attA\leq a_i$ and
$f_j$ has condition $\attA\leq a_j$,
then $a_i > a_j$.

\begin{theorem}
    The $\bdI(\famop{\leq}{\gets})$ problem
    can be solved in $O(N\log N)$ time.
\label{leqgets}
\end{theorem}

\begin{theorem}
    The $\bdI(\famop{\leq}{+})$ problem
    can be solved in $O(N\log N)$ time.
\label{leqadd}
\end{theorem}

\begin{theorem}
    The $\bdI(\famop{\leq}{\getsadd})$ problem
    can be solved in $O(N\log N)$ time.
\label{leqgetsadd}
\end{theorem}

\begin{theorem}
    The $\bdI(\famop{\leq}{\affine})$ problem
    can be solved in $O(N\log N)$ time.
\label{leqaffine}
\end{theorem}

\subsection{With At-most/At-least Conditions}
With \emph{at-most/at-least} ($\attA\leq a$ \textbf{or} $\attA\geq a$) conditions, the arguments
from the previous section cannot be directly reused.
In fact, the introduction of this new condition type is where
we first encounter the hardness boundary for most families
of operations.

\subsubsection{With Assignment Modifiers}
There is still a polynomial time algorithm for the family of operations
with the \emph{assignment} modifier, following the reasoning that
it is possible to avoid having a tuple selected by both an \emph{at-most}
condition and an \emph{at-least} condition.

\begin{theorem}
    The $\bdI(\famop{\lgeq}{\gets})$ problem
    can be solved in $O(N^2)$ time.
\label{lgeqgets}
\end{theorem}

\subsubsection{With Increment Modifiers}
This is the first time we encounter the hardness boundary.
Despite the fact that the operations are commutative, we cannot
utilize the same techniques as we did for other families of operations.

\begin{theorem}
    The $\bdI(\famop{\lgeq}{+})$ problem
    is $\NP$-hard.
\label{lgeqadd}
\end{theorem}
In order to prove Theorem~\ref{lgeqadd},
we provide a polynomial-time reduction from \textsc{SubsetSum},
which is a known $\NP$-hard problem, defined as follows~\cite{gareyjohnson}.
\begin{definition}[SubsetSum]
    The \textsc{SubsetSum} decision problem is, given a set $S=\{s_1,\ldots,s_n\}$ of
    positive integers, and a positive integer $t$, determine whether
    there exists a subset $T\subseteq S$ such that the sum of all elements in $T$
    equals $t$.
\end{definition}
Consider an instance of the \textsc{SubsetSum} problem with a set
$S=\{s_1,\ldots,s_n\}$ of positive integers and a positive integer $t$.
The reduction is as follows: let $s_0=-t$ and
\begin{align*}
    \Rbeg &= \left\{(\attK=k, \attA=k, \attB=0)  \mid k\in\{0,\ldots,n\}\right\}\\
    \Rend &= \left\{(\attK=k, \attA=k, \attB=b_k)\mid k\in\{0,\ldots,n\}\right\}
\end{align*}
where $b_k=\sum_{\ell=0}^k s_\ell$.
This reduction takes polynomial time. The claim is that it is
a positive instance of \textsc{SubsetSum} if and only if the
best diff between $\Rbeg$ and $\Rend$ under $\famop{\lgeq}{+}$
has cost $n$.
We show the correctness of this reduction via a series of lemmas.
Throughout this subsection, $\Rbeg$ and $\Rend$ refer to the sets
of tuples from the reduction as described here.

\begin{lemma}
    Operations in $\famop{\lgeq}{+}$ are commutative.
\label{lgeqcommute}
\end{lemma}
\begin{proof}
    This follows immediately from commutativity of addition and the fact that
    attributes in $\calA$ never change as a result of an $(\calA,\calB)$-operation.
\end{proof}

\begin{lemma}
    $\diff(\Rbeg,\Rend,\famop{\lgeq}{+})$ is nonempty,
    and if $F$ is its best diff,
    then $\cost(F)\leq n+1$.
\label{lgequpper}
\end{lemma}
\begin{proof}
    A sequence of operations $F=(f_0,\ldots,f_n)$ where
    $f_k=(\attA\geq k, \attB\gets\attB+s_k)$ for $k\in\{0,\ldots,n\}$
    is a diff between $\Rbeg$ and $\Rend$, and $\cost(F)=n+1$.
\end{proof}

Next, we establish a few lemmas claiming that there must be best diffs
between $\Rbeg$ and $\Rend$ satisfying certain properties.

\begin{lemma}
    $\diff(\Rbeg,\Rend,\famop{\lgeq}{+})$
    contains a \emph{bounded} best diff $F'=(f'_1,\ldots,f'_m)$,
    in which for all $i\in[m]$,
    \begin{align*}
    f'_i &=(\attA\leq a'_i,\attB\gets\attB+b'_i)\text{ or}\\
    f'_i &=(\attA\geq a'_i,\attB\gets\attB+b'_i)
    \end{align*}
    where $a'_i$ is an integer in $\{0,\ldots,n\}$.
\label{lgeqinteger}
\end{lemma}
\begin{proof}
    By Lemma~\ref{lgequpper},
    $\diff(\Rbeg,\Rend,\famop{\lgeq}{+})$
    contains a best diff $F=(f_1,\ldots,f_m)$.
    Entries in the $\attA$ attribute in $\Rbeg$ and $\Rend$, by construction,
    are integers in $\{0,\ldots,n\}$. Define
    \[\mathrm{bnd}(a)=\max\{0,\min\{n,a\}\}\]
    We construct $F'=(f'_1,\ldots,f'_m)$ from $F$:
    for each $i\in[m]$,
    \begin{itemize}
    \item if $f_i=(\attA\leq a_i,\attB\gets\attB+b_i)$,
    then we construct $f'_i=(\attA\leq \lfloor\mathrm{bnd}(a_i)\rfloor,\attB\gets\attB+b_i)$,
    since $\attA\leq a_i$ if and only if $\attA\leq\lfloor\mathrm{bnd}(a_i)\rfloor$.
    \item if $f_i=(\attA\geq a_i,\attB\gets\attB+b_i)$,
    then we construct $f'_i=(\attA\geq \lceil\mathrm{bnd}(a_i)\rceil,\attB\gets\attB+b_i)$,
    since $\attA\geq a_i$ if and only if $\attA\geq\lceil\mathrm{bnd}(a_i)\rceil$.
    \end{itemize}
    Thus, $F'$ is a bounded best diff in
    $\diff(\Rbeg,\Rend,\famop{\lgeq}{+})$.
\end{proof}

\begin{definition}[Gap operation]
    A \emph{gap operation at $k$}, where $k\in[n]$, is an
    operation with the condition $\attA\leq k-1$
    or the condition $\attA\geq k$.
\end{definition}

\begin{lemma}
    $\diff(\Rbeg,\Rend,\famop{\lgeq}{+})$
    contains a \emph{canonical} best diff $F'$
    where $F'$ contains
    exactly one gap operation at $k$, which must
    be either
    \begin{align*}
        f&=(\attA\leq k-1, \attB\gets \attB-s_k)\text{ or}\\
        f&=(\attA\geq k, \attB\gets \attB+s_k)
    \end{align*}
    for every $k\in [n]$,
\label{lgeqseparation}
\label{lgeqcanonical}
\end{lemma}
\begin{proof}
    Let $F=(f_1,\ldots,f_m)$ be the bounded best diff with the fewest gap operations.

    First, we prove that $F$ has at most $n$ gap operations, one at every $k\in[n]$.
    The proof follows. If $F$ contains two gap operations with the same condition, by Lemma~\ref{lgeqcommute},
    they can be reordered and combined, reducing the number of gap operations, a contradiction.
    If $F$ contains both
    \begin{align*}
        f_i&=(\attA\leq k-1, \attB\gets \attB+b_i)\text{ and}\\
        f_j&=(\attA\geq k, \attB\gets \attB+b_j)
    \end{align*}
    for some $k\in[n]$, then we
    can replace them with
    \begin{align*}
        g_1&=(\attA\leq n, \attB\gets \attB+b_i)\text{ and}\\
        g_2&=(\attA\geq k, \attB\gets \attB+(b_j-b_i))
    \end{align*}
    to obtain a best diff
    with one fewer gap operation ($g_1$ is not a gap operation), a contradiction.
    \begin{center}
    \begin{tikzpicture}
        \begin{scope}[xshift=0.0cm]
            \draw[thin] (-0.2,-1.5) rectangle (3.7,0.3);
            \draw[<->,thick] (-0.1,0)--(3.1,0) node[anchor=west] {$\attA$};
            \draw[<-|] (0,-0.3)--node[below=-0.5pt]{$\attB\gets \attB+b_i$} (1.3,-0.3) node[anchor=west] {$k-1$};
            \draw[|->] (1.6,-1.0) node[anchor=east] {$k$} --node[below=-0.5pt]{$\attB\gets \attB+b_j$} (3.0,-1.0);
        \end{scope}
        \node () at (4.0,-0.6) {$=$};
        \begin{scope}[xshift=4.5cm]
            \draw[thin] (-0.2,-1.5) rectangle (3.7,0.3);
            \draw[<->,thick] (-0.1,0)--(3.1,0) node[anchor=west] {$\attA$};
            \draw[<->] (0,-0.3)--node[below=-1pt]{$\attB\gets \attB+b_i$} (3,-0.3);
            \draw[|->] (1.6,-1.0) node[anchor=east] {$k$} --node[below=-1pt]{$\attB\gets \attB+(b_j-b_i)$} (3.0,-1.0);
        \end{scope}
    \end{tikzpicture}
    \end{center}

    Second, we prove that $F$ has at least $n$ gap operations, one at every $k\in[n]$.
    The proof follows. Assume that there is a value $k\in[n]$ such that
    $F$ has no gap operation at $k$; that is,
    all operations in $F$ has neither the condition $\attA\leq k-1$ nor the condition
    $\attA\geq k$. One can show by induction on the number of operations
    performed on $\Rbeg$ that the tuple with $\attK=\attA=k-1$
    and the tuple with $\attK=\attA=k$ will always
    have the same value in the $\attB$ attribute.
    More precisely, for any $i\in\{0,\ldots,m\}$
    and $F_i=(f_1,\ldots,f_i)$, in the relation $F_i(\Rbeg)$,
    the tuple with $\attK=\attA=k-1$
    and the tuple with $\attK=\attA=k$ have the same value in the $\attB$ attribute.
    However, in $\Rend$ those values in the $\attB$ attribute
    differ by $s_k\neq 0$ by construction, a contradiction.

    Through a similar argument, for each $k\in[n]$, the gap operation
    at $k$ in $F$ must be
    either
    \begin{align*}
        f&=(\attA\leq k-1, \attB\gets \attB-s_k)\text{ or}\\
        f&=(\attA\geq k, \attB\gets \attB+s_k)
    \end{align*}
    for otherwise the difference between the $\attB$ values of
    the tuple with $\attK=\attA=k-1$
    and the tuple with $\attK=\attA=k$ in $F(\Rbeg)$ will not be $s_k$,
    which implies that $F(\Rbeg)\neq \Rend$.
\end{proof}

\begin{lemma}
    If $F$ is a best diff in
    $\diff(\Rbeg,\Rend,\famop{\lgeq}{+})$,
    then $\cost(F)\geq n$.
\label{lgeqlower}
\end{lemma}
\begin{proof}
    This is a corollary of Lemma~\ref{lgeqseparation}.
\end{proof}

\begin{lemma}
    Best diffs in $\diff(\Rbeg,\Rend,\famop{\lgeq}{+})$
    have cost $n$ if and only if
    there is a subset $T\subseteq S$ such that the sum of all elements in $T$ equals $t$.
\label{lgeqreduction}
\end{lemma}
\begin{proof}
    ($\Leftarrow$) Let $T$ be a subset of $S$ such that
    the sum of elements in $T$ equals $t$, then $F=(f_1,\ldots,f_n)$ where,
    for each $k\in[n]$,
    \[f_k=\begin{cases}
    (\attA\leq k-1, \attB\gets\attB-s_k)&\text{if $s_k\in T$}\\
    (\attA\geq k, \attB\gets\attB+s_k)&\text{otherwise}
    \end{cases}\]
    is a best diff in
    $\diff(\Rbeg,\Rend,\famop{\lgeq}{+})$
     with cost $n$.

    ($\Rightarrow$) By Lemma~\ref{lgeqinteger} and Lemma~\ref{lgeqcanonical},
    $\diff(\Rbeg,\Rend,\famop{\lgeq}{+})$ contains
    a canonical best diff $F=(f_1,\ldots,f_n)$
    with cost $n$.

    Because the cost is $n$, $F$ contains exactly one gap operation at $k$
    for each $k\in[n]$, as described in Lemma~\ref{lgeqcanonical} and nothing else.
    By Lemma~\ref{lgeqcommute} let $f_k$ be the gap operation at $k$
    for each $k\in[n]$.
    If $f_k$ has the condition $\attA\geq k$, it does not affect the tuple
    with $\attK=\attA=0$. Let $\{f_{i_1},\ldots,f_{i_m}\}$ be the subset of
    $\{f_1,\ldots,f_n\}$ of operations whose conditions are of the form
    $\attA\leq i_\ell-1$. Thus, in the relation $F(\Rbeg)=\Rend$, the tuple with $\attK=\attA=0$
    has the value in attribute $\attB$ equal to $-\sum_{\ell=1}^m{s_{i_\ell}}=-t$. Thus,
    $T=\{s_{i_1},\ldots,s_{i_m}\}$ is a subset of $S$ whose sum of elements is equal
    to $t$.
\end{proof}

This proves the correctness of the polynomial-time reduction
from \textsc{SubsetSum}, which concludes the $\NP$-hardness proof
for Theorem~\ref{lgeqadd}.

\begin{figure}[!hbt]
\centering
\begin{tikzpicture}
    \draw[<->,thick] (-0.1,0)--(7.1,0) node[anchor=west] {$\attA$};
    \foreach \i/\j in {0/-93,1/-92,2/-89,3/-80,4/-53,5/28} {
        \node () at (1.5+\i*0.8,1.2) {$0$};
        \node () at (1.5+\i*0.8,0.8) {$\j$};
        \node () at (1.5+\i*0.8,0.3) {$\i$};
        \fill (1.5+\i*0.8,0.0) circle (0.08cm);
    }
    \node () at (6.7,1.2) {$\attB$ in $\Rbeg$};
    \node () at (6.7,0.8) {$\attB$ in $\Rend$};
    \draw[<-|,dotted] (0,-0.3)--node[below=-0.5pt]{$\attB\gets \attB-1$} (1.5,-0.3);
    \draw[|->] (2.3,-0.3)--node[below=-0.5pt]{$\attB\gets \attB+1$} (7.0,-0.3);
    \draw[<-|] (0,-1.0)--node[below=-0.5pt]{$\attB\gets \attB-3$} (2.3,-1.0);
    \draw[|->,dotted] (3.1,-1.0)--node[below=-0.5pt]{$\attB\gets \attB+3$} (7.0,-1.0);
    \draw[<-|] (0,-1.7)--node[below=-0.5pt]{$\attB\gets \attB-9$} (3.1,-1.7);
    \draw[|->,dotted] (3.9,-1.7)--node[below=-0.5pt]{$\attB\gets \attB+9$} (7.0,-1.7);
    \draw[<-|,dotted] (0,-2.4)--node[below=-0.5pt]{$\attB\gets \attB-27$} (3.9,-2.4);
    \draw[|->] (4.7,-2.4)--node[below=-0.5pt]{$\attB\gets \attB+27$} (7.0,-2.4);
    \draw[<-|] (0,-3.1)--node[below=-0.5pt]{$\attB\gets \attB-81$} (4.7,-3.1);
    \draw[|->,dotted] (5.5,-3.1)--node[below=-0.5pt]{$\attB\gets \attB+81$} (7.0,-3.1);
\end{tikzpicture}
\caption{Illustration of Example \ref{subsetsumexample}}
\label{fig:subsetsumillus}
\end{figure}
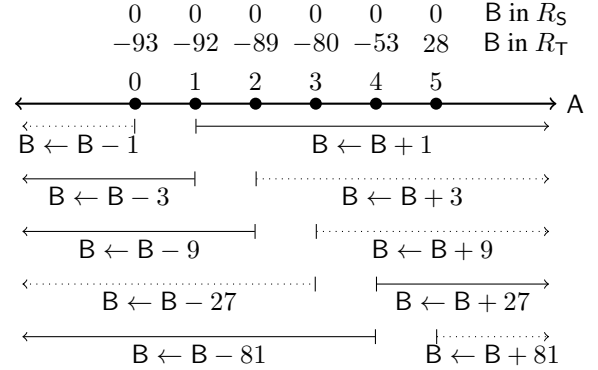

\begin{example}
\label{subsetsumexample}
Consider the \textsc{SubsetSum} instance with
$S=\{1,3,9,27,81\}$ and $t=93$.
The subset $T=\{3,9,81\}$ of $S$
has the sum of its elements equal to $t$. The reduction
gives the following instance of the
$\bdI(\famop{\lgeq}{+})$ problem.
\begin{center}
\begin{tabular}{c c c c c}
$\attK$&$\attA$&$\attB$\\
\hline
$0$&$0$&$0$\\
$1$&$1$&$0$\\
$2$&$2$&$0$\\
$3$&$3$&$0$\\
$4$&$4$&$0$\\
$5$&$5$&$0$
\end{tabular}
\begin{tikzpicture}
\path[draw=black,solid,line width=1mm,fill=black,
    preaction={-triangle 90,thin,draw,shorten >=-1mm}
] (0, -0.3) -- (0.3, -0.3);
\end{tikzpicture}
\begin{tabular}{c c c c c}
$\attK$&$\attA$&$\attB$\\
\hline
$0$&$0$&$-93$\\
$1$&$1$&$-92$\\
$2$&$2$&$-89$\\
$3$&$3$&$-80$\\
$4$&$4$&$-53$\\
$5$&$5$&$28$
\end{tabular}
\end{center}
Figure~\ref{fig:subsetsumillus} shows the two possible gap operations at $k$
for each $k\in[n]$ in their own row. In the third row, for example, one
of the two gap operations modifying $\attB$ by $9$ must be used to ensure
that the $\attB$ values of the middle tuples differ by $9$
(in the final relation, between $-89$ and $-80$).
In this case, $F=(f_1,f_2,f_3,f_4,f_5)$ where
\begin{align*}
f_1 &= (\attA\geq 1,\attB\gets\attB+1)\\
f_2 &= (\attA\leq 1,\attB\gets\attB-3)\\
f_3 &= (\attA\leq 2,\attB\gets\attB-9)\\
f_4 &= (\attA\geq 4,\attB\gets\attB+27)\\
f_5 &= (\attA\leq 4,\attB\gets\attB-81)
\end{align*}
is a best diff with cost $5$. The corresponding chosen gap operations
are shown in solid lines, while the ones not chosen are shown in dotted lines.
\end{example}

\subsubsection{With Assignment/Increment or Affine Modifiers}
With \emph{assignment/increment} or \emph{affine} modifiers,
the problem is still $\NP$-hard.

For the \emph{assignment/increment} modifiers, this can be shown via an extension
of the proof above for the version with only \emph{increment} modifiers.
Essentially, the proof is to show that the assignment modifier does not
provide additional expressivity in the reduction given.

\begin{theorem}
    The $\bdI(\famop{\lgeq}{\getsadd})$ problem
    is $\NP$-hard.
\label{lgeqgetsadd}
\end{theorem}

We prove the aforementioned theorem via the following lemma.
\begin{lemma}
    Best diffs in $\diff(\Rbeg,\Rend,\famop{\lgeq}{\getsadd})$
    have cost $n$ if and only if best diffs in
    $\diff(\Rbeg,\Rend,\famop{\lgeq}{+})$
    have cost $n$.
\label{lgeqaddandgetsadd}
\end{lemma}
\begin{proof}
    ($\Leftarrow$) Any diff in $\diff(\Rbeg,\Rend,\famop{\lgeq}{+})$
    is also a diff in $\diff(\Rbeg,\Rend,\famop{\lgeq}{\getsadd})$.

    ($\Rightarrow$)
    Let $F=(f_1,\ldots,f_n)$ be a best diff in
    $\diff(\Rbeg,\Rend,\famop{\lgeq}{\getsadd})$
    of cost $n$ that has the smallest number of assignment modifiers
    and, among the best diffs with the smallest number of assignment modifiers,
    has the smallest total length. We show that $F$ has no
    assignment modifiers.

    The proof follows.
    Assume to the contrary, and let $i$ be the smallest index in $[n]$ such that
    $f_i$ has an assignment modifier. Suppose $f_i=(\attA\leq a_i, \attB\gets b_i)$.
    (The proof for when $f_i$ has condition $\attA\geq a_i$ is similar.)

    Case 1: There is an operation $f_j=(\attA\geq a_j,\attB\gets\attB+b_j)$
    where $j<i$ and $a_j\leq a_i$. Then, let
    $f'_j=(\attA\geq a_i+1,\attB\gets\attB+b_j)$. If $F'$ is defined as
    $F$ where $f_j$ is replaced with $f'_j$, then $F'$ would still yield
    $F'(\Rbeg)=\Rend$, but the total length of $F'$ is smaller than that of $F$.

    Case 2: There is an operation $f_j=(\attA\leq a_j,\attB\gets\attB+b_j)$
    where $j<i$ and $a_j\leq a_i$. If $F'$ is defined as
    $F$ where $f_j$ is removed, then $F'$ would still yield
    $F'(\Rbeg)=\Rend$, but the cost of $F'$ is smaller than that of $F$.

    Case 3: None of the above. Then, all tuples matching $\attA\leq a_i$
    still have the same value in the $\attB$ attribute, say $\beta$, in
    $F''(\Rbeg)$ where $F''=(f_1,\ldots,f_{i-1})$. Then, let
    $f'_i=(\attA\leq a_i,\attB\gets\attB+(b_i-\beta))$. If $F'$ is defined as
    $F$ where $f_i$ is replaced with $f'_i$, then $F'$ would still yield
    $F'(\Rbeg)=\Rend$, but $F'$ has fewer assignment modifiers than $F$.

    Therefore, $F$ has no assignment modifiers. Thus, $F$ is also
    a best diff in $\diff(\Rbeg,\Rend,\famop{\lgeq}{+})$.
\end{proof}

With the \emph{affine} modifier, we again show $\NP$-hardness
via a polynomial-time reduction from \textsc{SubsetSum}, but
the reduction is slightly different from the \emph{increment} case.

\begin{theorem}
    The $\bdI(\famop{\lgeq}{\affine})$ problem
    is $\NP$-hard.
\label{lgeqaffine}
\end{theorem}

Consider an instance of the \textsc{SubsetSum} problem with a set
$S=\{s_1,\ldots,s_n\}$ of positive integers and a positive integer $t$.
The reduction is as follows: let $s_0=-t$ and
\begin{align*}
    \Rbeg &= \bigcup_{k\in\{0,\ldots,n\}}\left\{(\attK=\attA=99k+i, \attB=0)  \mid i\in[99]\right\}\\
    \Rend &= \bigcup_{k\in\{0,\ldots,n\}}\left\{(\attK=\attA=99k+i, \attB=b_k)\mid i\in[99]\right\}
\end{align*}
where $b_k=\sum_{\ell=0}^k s_\ell$.
This reduction takes polynomial time. The claim is that it is
a positive instance of \textsc{SubsetSum} if and only if the
best diff between $\Rbeg$ and $\Rend$ under $\famop{\lgeq}{\affine}$
has cost $n$.

The proof is similar to that given for \emph{increment}
and \emph{assignment/increment}, and thus only the differences are sketched here.
In the reduction, instead of one tuple for
each integer in $S$, a \emph{block} of $99$ tuples with the same $\attB$
value is created.
Intuitively, if an operation has a modifier with nonzero slope
($\attB\gets b\attB+c$ with $b\neq 0$)
and it matches multiple tuples in the same block,
then it can break the ``same $\attB$ value''
requirement within that block. It can take a few operations or one
operation with zero slope to fix the block. It can be shown that modifiers
with nonzero slope are unnecessary in the best diff in this instance.

\subsection{With Range Conditions}
With \emph{range} ($\attA\in [a,z]$) conditions, the problem is $\NP$-hard
for all families of operations of interest, except the one with
the \emph{assignment} modifier, similar to the previous case
with \emph{at-most/at-least} conditions. The arguments
utilize the same core ideas, but are somewhat more complicated.

\subsubsection{With Assignment Modifiers}
As in cases previously discussed, there
is a polynomial time algorithm for the family of operations
with the \emph{assignment} modifier. The reasoning is slightly
different although the main idea is similar: it is possible to avoid having a tuple
selected by two ranges that partially overlap. That is, there is a diff
for which any two ranges are either completely disjoint or are
such that one is completely contained within the other.

\begin{theorem}
    The $\bdI(\famop{\range}{\gets})$ problem
    can be solved in $O(N^4)$ time.
\label{rangegets}
\end{theorem}

\subsubsection{With Increment Modifiers}
With \emph{increment} modifiers, like before, the problem is $\NP$-hard.
This follows from the same reduction from \textsc{SubsetSum} given
in the proof of Theorem~\ref{lgeqadd}. The proof of the reduction's correctness,
however, is somewhat different.

\begin{theorem}
    The $\bdI(\famop{\range}{+})$ problem
    is $\NP$-hard.
\label{rangeadd}
\end{theorem}

We prove Theorem~\ref{rangeadd} via a series of lemmas.
Throughout this subsection, $\Rbeg$ and $\Rend$ refer to the sets
of tuples from the reduction.

\begin{lemma}
Operations in $\famop{\range}{+}$ are commutative.
\label{rangecommute}
\end{lemma}
\begin{proof}
    same as in \ref{lgeqcommute}
\end{proof}

\begin{lemma}
    $\diff(\Rbeg,\Rend,\famop{\range}{+})$ is nonempty,
    and if $F$ is its best diff,
    then $\cost(F)\leq n+1$.
\label{rangeupper}
\end{lemma}
\begin{proof}
    A sequence of operations $F=(f_0,\ldots,f_n)$ where
    $f_k=(\attA\in[k,n], \attB\gets\attB+s_k)$ for $k\in\{0,\ldots,n\}$
    is a diff between $\Rbeg$ and $\Rend$, and $\cost(F)=n+1$.
\end{proof}

\begin{lemma}
    $\diff(\Rbeg,\Rend,\famop{\range}{+})$
    contains a
    \emph{bounded} best diff $F'=(f'_1,\ldots,f'_m)$
    in which for all $i\in[m]$,
    $f'_i=(\attA\in [a'_i,z'_i],\attB\gets\attB+b'_i)$ where
    $a'_i$ and $z'_i$ are integers in $\{0,\ldots,n\}$.
\label{rangeinteger}
\end{lemma}
\begin{proof}
    By Lemma~\ref{rangeupper},
    $\diff(\Rbeg,\Rend,\famop{\range}{+})$
    contains a best diff $F=(f_1,\ldots,f_m)$.
    Entries in the $\attA$ attribute in $\Rbeg$ and $\Rend$, by construction,
    are integers in $\{0,\ldots,n\}$. Define
    \[\mathrm{bnd}(a)=\max\{0,\min\{n,a\}\}\]
    We construct $F'=(f'_1,\ldots,f'_m)$ from $F$:
    for each $i\in[m]$,
    \begin{itemize}
    \item if $f_i=(\attA\in[a_i,z_i],\attB\gets\attB+b_i)$,
    then we construct $f'_i=(\attA\in [\lceil\mathrm{bnd}(a_i)\rceil,\lfloor\mathrm{bnd}(z_i)\rfloor],\attB\gets\attB+b_i)$,
    since $\attA\in[a_i,z_i]$ if and only if $\attA\in[\lceil\mathrm{bnd}(a_i)\rceil,\lfloor\mathrm{bnd}(z_i)\rfloor]$.
    \end{itemize}
    Thus, $F'$ is a bounded best diff in
    $\diff(\Rbeg,\Rend,\famop{\range}{+})$.
\end{proof}

\begin{lemma}
    If $F$ is a best diff in
    $\diff(\Rbeg,\Rend,\famop{\range}{+})$,
    then $\cost(F)\geq n$.
\label{rangelower}
\end{lemma}
\begin{proof}
    Define the \emph{jump} of a relation as the number of values $i\in[n]$ such that for tuples
    $(\attK=i-1, \attA=i-1, \attB=b_{i-1})$
    and $(\attK=i, \attA=i, \attB=b_{i})$, we have $b_{i-1}<b_i$.
    Note that the jumps in $\Rbeg$ and $\Rend$ are $0$ and $n$, respectively.
    We prove the following statement by induction: after applying
    $m$ operations from $\famop{\range}{+}$ to $\Rbeg$, the jump
    of the resulting relation is at most $m$. This implies that at least
    $n$ operations are required to transform $\Rbeg$ into $\Rend$.

    The proof follows.
    The base case $m=0$ is trivial. Assume, as an induction hypothesis,
     that for $m'<m$,
    applying $m'$ operations to $\Rbeg$ resulting in jump that is at most $m'$.
    Let $R'$ be the result of applying $m-1$ operations
    on $\Rbeg$, and thus its jump is at most $m-1$.
    Consider applying $f=(\attA\in[a, z], \attB\gets\attB+b)$ to $R'$
    and let $f(R')=R''$.
    Consider tuples $(\attK=i-1, \attA=i-1, \attB=b_{i-1})$
    and $(\attK=i, \attA=i, \attB=b_{i})$ in $R'$ where $b_{i-1}\geq b_i$.
    \begin{itemize}
    \item If $i-1$ and $i$ are both not in $[a, z]$, then the $\attB$ values
    remain $b_{i-1}$ and $b_i$ respectively, and $b_{i-1}\geq b_i$.
    This does not contribute to increase in jump.
    \item If $i-1$ and $i$ are both in $[a, z]$, then the $\attB$ values
    become $b_{i-1}+b$ and $b_i+b$ respectively, and $b_{i-1}+b\geq b_i+b$.
    This does not contribute to increase in jump.
    \item If $i-1 < a \leq i$, then the $\attB$ values
    become $b_{i-1}$ and $b_i+b$ respectively, and if $b_{i-1}< b_i+b$,
    then $b>0$.
    \item If $i-1 \leq z < i$, then the $\attB$ values
    become $b_{i-1}+b$ and $b_i$ respectively, and if $b_{i-1}+b< b_i$,
    then $b<0$.
    \end{itemize}
    Thus, jump can only increase by at most $1$ depending on the value of $b$:
    if $b>0$, then jump can only increase because of $i$ where $i-1<a\leq i$,
    and if $b<0$, then jump can only increase because of $i$ where $i-1\leq z<i$.
\end{proof}

\begin{lemma}
    $\diff(\Rbeg,\Rend,\famop{\range}{+})$
    contains a bounded best diff
    $F=(f_1,\ldots,f_m)$ in which there are no two
    operations
    \begin{align*}
        f_i&=(\attA\in[a_i,z_i],\attB\gets\attB+b_i)\text{ and}\\
        f_j&=(\attA\in[a_j,z_j],\attB\gets\attB+b_j)
    \end{align*}
    such that $a_i=a_j$ or $z_i=z_j$.
\label{rangenocollision}
\end{lemma}
\begin{proof}
    Define a \emph{collision} of $F$ as a pair $(i,j)$
    where $i,j\in[m]$ and $i<j$ such that $a_i=a_j$ or $z_i=z_j$.

    By Lemma~\ref{rangeinteger}, let $F=(f_1,\ldots,f_m)$,
    where $f_i=(\attA\in[a_i,z_i],\attB\gets\attB+b_i)$ for all $i\in[m]$,
    be a bounded best diff with the smallest total length.
    We show that $F$ contains no collisions.

    The proof follows.
    Assume to the contrary that $F$ has a collision $(i,j)$.
    Suppose $a_i=a_j$. (The argument for when $z_i=z_j$ is symmetrical.)
    By commutativity,
    \[F'=(f_1,\ldots,f_{i-1},f_i,f_j,f_{i+1},\ldots,f_{j-1},f_{j+1},\ldots,f_m)\]
    is also a bounded best diff.

    Case 1: if $z_i = z_j$ then let
    \[    g  = (\attA\in[a_i,z_i],\attB\gets\attB+(b_i+b_j)) \]
    then $F''$ defined as follows is also a bounded best diff:
    \[F''=(f_1,\ldots,f_{i-1},g,f_{i+1},\ldots,f_{j-1},f_{j+1},\ldots,f_m)\]
    \begin{center}
    \begin{tikzpicture}
        \begin{scope}[xshift=0.0cm]
            \draw[thin] (-0.2,-1.5) rectangle (3.7,0.3);
            \draw[<->,thick] (-0.1,0)--(3.1,0) node[anchor=west] {$\attA$};
            \draw[|-|] (0.5,-0.3) node[anchor=east] {$a_i$} --node[below=-1.5pt]{$\attB\gets \attB+b_i$} (2.5,-0.3) node[anchor=west] {$z_i$};
            \draw[|-|] (0.5,-1.0) node[anchor=east] {$a_j$} --node[below=-1.5pt]{$\attB\gets \attB+b_j$} (2.5,-1.0) node[anchor=west] {$z_j$};
        \end{scope}
        \node () at (4.0,-0.6) {$=$};
        \begin{scope}[xshift=4.5cm]
            \draw[thin] (-0.2,-1.5) rectangle (3.7,0.3);
            \draw[<->,thick] (-0.1,0)--(3.1,0) node[anchor=west] {$\attA$};
            \draw[|-|] (0.5,-0.65) node[anchor=east] {$a_i$} --node[below=-0.5pt]{$\attB\gets \attB+(b_i+b_j)$} (2.5,-0.65) node[anchor=west] {$z_i$};
        \end{scope}
    \end{tikzpicture}
    \end{center}
    However, $F''$ has smaller cost than $F$,
    contradicting the fact that $F$ is a best diff.

    Case 2: if $z_i < z_j$ then let
    \begin{align*}
        g_1 &= (\attA\in[a_i,z_i],\attB\gets\attB+(b_i+b_j))\\
        g_2 &= (\attA\in[z_i+1,z_j],\attB\gets\attB+b_j)
    \end{align*}
    then $F''$ defined as follows is also a bounded best diff:
    \[F''=(f_1,\ldots,f_{i-1},g_1,g_2,f_{i+1},\ldots,f_{j-1},f_{j+1},\ldots,f_m)\]
    \begin{center}
    \begin{tikzpicture}
        \begin{scope}[xshift=0.0cm]
            \draw[thin] (-0.2,-1.5) rectangle (3.7,0.3);
            \draw[<->,thick] (-0.1,0)--(3.1,0) node[anchor=west] {$\attA$};
            \draw[|-|] (0.3,-0.3) node[anchor=east] {$a_i$} --node[below=-1.5pt]{$\attB\gets \attB+b_i$} (1.7,-0.3) node[anchor=west] {$z_i$};
            \draw[|-|] (0.3,-1.0) node[anchor=east] {$a_j$} --node[below=-1.5pt]{$\attB\gets \attB+b_j$} (3.0,-1.0) node[anchor=west] {$z_j$};
        \end{scope}
        \node () at (4.0,-0.6) {$=$};
        \begin{scope}[xshift=4.5cm]
            \draw[thin] (-0.2,-1.5) rectangle (3.7,0.3);
            \draw[<->,thick] (-0.1,0)--(3.1,0) node[anchor=west] {$\attA$};
            \draw[|-|] (0.3,-0.3) node[anchor=east] {$a_i$} --node[below=-1.5pt,xshift=5pt]{$\attB\gets \attB+(b_i+b_j)$} (1.7,-0.3) node[anchor=west] {$z_i$};
            \draw[|-|] (1.8,-1.0) node[anchor=east] {$z_i+1$} --node[below=-0.5pt,xshift=2pt]{$\attB\gets \attB+b_j$} (3.0,-1.0) node[anchor=west] {$z_j$};
        \end{scope}
    \end{tikzpicture}
    \end{center}
    However, $\ell(F'')=\ell(F')-(z_i-a_i+1)<\ell(F')=\ell(F)$,
    contradicting the fact that $F'$ has the smallest total length.

    Case 3: if $z_i > z_j$, the proof is similar to Case 2.

    Therefore, $F$ has no collisions,
    and thus $\diff(\Rbeg,\Rend,\famop{\range}{+})$
    contains a bounded best diff
    $F=(f_1,\ldots,f_m)$ in which there are no two
    operations
    \begin{align*}
        f_i&=(\attA\in[a_i,z_i],\attB\gets\attB+b_i)\text{ and}\\
        f_j&=(\attA\in[a_j,z_j],\attB\gets\attB+b_j)
    \end{align*}
    such that $a_i=a_j$ or $z_i=z_j$.
\end{proof}

\begin{lemma}
    Best diffs in $\diff(\Rbeg,\Rend,\famop{\range}{+})$
    have cost $n$ if and only if best diffs in
    $\diff(\Rbeg,\Rend,\famop{\lgeq}{+})$
    have cost $n$.
\label{lgeqandrange}
\end{lemma}
The idea of the proof is that a diff from one set can be translated
into a diff from the other set with the same cost. The full proof is
given in Appendix~\ref{sec:additionalproofs}.

By Lemma~\ref{lgeqreduction} and Lemma~\ref{lgeqandrange},
the reduction is correct, implying Theorem~\ref{rangeadd}.

\subsubsection{With Assignment/Increment or Affine Modifiers}
With \emph{assignment/increment} or \emph{affine} modifiers,
once again, the problem is still $\NP$-hard.

\begin{theorem}
    The $\bdI(\famop{\range}{\getsadd})$ problem
    is $\NP$-hard.
\label{rangegetsadd}
\end{theorem}

The proof of the theorem is still based on the same
reduction from \textsc{SubsetSum}, and follows from the following lemma,
the proof of which is given in Appendix~\ref{sec:additionalproofs}.
\begin{lemma}
    Best diffs in $\diff(\Rbeg,\Rend,\famop{\range}{\getsadd})$
    have cost $n$ if and only if best diffs in
    $\diff(\Rbeg,\Rend,\famop{\range}{+})$
    have cost $n$.
\label{rangeaddandgetsadd}
\end{lemma}

With the \emph{affine} modifier, $\NP$-hardness can be shown
using the same polynomial-time reduction from \textsc{SubsetSum}
as given for Theorem~\ref{lgeqaffine}.

\begin{theorem}
    The $\bdI(\famop{\range}{\affine})$ problem
    is $\NP$-hard.
\label{rangeaffine}
\end{theorem}

\subsection{With Union-of-Ranges Conditions}
With the \emph{union-of-ranges} conditions, the problem becomes
$\NP$-hard even with the \emph{assignment} modifier.

\begin{theorem}
    The $\bdI(\famop{\rangeunion}{\gets})$ problem
    is $\NP$-hard.
\label{uniongets}
\end{theorem}

In order to prove Theorem~\ref{uniongets},
we provide a polynomial-time reduction from \textsc{2SCS}
(shortest common supersequence of strings of length two),
which is a known $\NP$-hard problem, defined as follows~\cite{timkovskii1989complexity}.
\begin{definition}[2SCS]
    The \textsc{2SCS} decision problem is, given a set $S=\{s_1,\ldots,s_n\}$ of
    strings of length two, and a nonnegative integer $t$, determine whether
    $S$ has a common supersequence of length at most $t$; that is,
    whether there exists a string $s$ of length at most $t$
    such that for each string $s_i\in S$, it is possible to remove
    some symbols (possibly none) from $s$ to obtain $s_i$.
\end{definition}

Note that the alphabet size is not necessarily constant:
there can be as many as $2n$ different symbols in a given instance.
Also, we assume that each symbol is given in the input
represented as a positive integer.

In fact, we will provide a polynomial-time reduction from
\textsc{2DistinctSCS}, which is similar to \textsc{2SCS} with
an additional restriction that the two letters in each string in $S$
are not the same. The proof that \textsc{2DistinctSCS} is $\NP$-hard,
via a reduction from \textsc{2SCS}, is given in Appendix~\ref{sec:2scs}.

Consider an instance of the \textsc{2DistinctSCS} problem with a set
$S=\{s_1,\ldots,s_n\}$ of strings of length two
and a nonnegative integer $t$.
For each $k\in[n]$, let $u_k$ and $v_k$ be
(positive integer representations of)
the two symbols of $s_k$ in order.
The reduction is as follows:
for $k\in[n]$, let
\begin{align*}
    t_{5k-3} &= (\attK=5k-3, \attA=4k-3, \attB=0)\\
    t_{5k-2} &= (\attK=5k-2, \attA=4k-2, \attB=0)\\
    t_{5k-1} &= (\attK=5k-1, \attA=4k-1, \attB=0)\\
    t'_{5k-3} &= (\attK=5k-3, \attA=4k-3, \attB=u_k)\\
    t'_{5k-2} &= (\attK=5k-2, \attA=4k-2, \attB=v_k)\\
    t'_{5k-1} &= (\attK=5k-1, \attA=4k-1, \attB=u_k)
\end{align*}
for $k\in\{0,\ldots,n\}$, let
\begin{align*}
    t_{5k+0} = t'_{5k+0} &= (\attK=5k+0, \attA=4k, \attB=-1)\\
    t_{5k+1} = t'_{5k+1} &= (\attK=5k+1, \attA=4k, \attB=-2)
\end{align*}
and let $\kappa_0=1$, $\kappa_1=t+99$, and
\begin{align*}
    \Rbeg &= \{t_\ell \mid\ell\in\{0,\ldots,5n+1\}\}\\
    \Rend &= \{t'_\ell\mid\ell\in\{0,\ldots,5n+1\}\}
\end{align*}
This reduction takes polynomial time. The claim is that it is
a positive instance of \textsc{2DistinctSCS} if and only if the
best diff between $\Rbeg$ and $\Rend$ under $\famop{\rangeunion}{\gets}$
has cost at most $t+2n(t+99)$.
We show the correctness of this reduction via a series of lemmas.
In this subsection, $\Rbeg$ and $\Rend$ refer to the sets
of tuples from the reduction as described here.

First, we define \emph{total range count}, which impacts
the cost.
\begin{definition}[Total Range Count]
    Let $F=(f_1,\ldots,f_m)$ where, for $i\in[m]$, $f_i\in\famop{\rangeunion}{\gets}$ and
    \[f_i=(\attA\in\bigcup_{j=1}^{r_i} [a_{ij},z_{ij}],\attB\gets b_i)\]
    The \emph{total range count} of $F$ is defined as $\sum_{i=1}^m r_i$.
\end{definition}
Next, we establish the special purpose of the tuples of the form
$\attA=4k$ in the construction, the proof of which is given in Appendix~\ref{sec:additionalproofs}.
\begin{lemma}
    A diff between $\Rbeg$ and $\Rend$ contains no operation
    whose condition matches $\attA=4k$ for any $k\in\{0,\ldots,n\}$.
\label{barrier}
\end{lemma}
These $\attA=4k$ tuples provide ``barriers'' over which no range condition
can cross. Thus, they break the possible $\attA$ values into partitions.
Let partition $k$, denoted $P_k$, refers to tuples whose $\attA$ value is
between $4k-4$ and $4k$, exclusive, for $k\in[n]$.
There are exactly three tuples in each partition. We say that an operation
\emph{affects} a partition if some tuple in that partition is matched by
the condition of the operation.
\begin{lemma}
    A diff between $\Rbeg$ and $\Rend$ has total range count at least $2n$,
    and each partition has at least two operations that affects it.
\label{rangecountlowerbound}
\end{lemma}
\begin{proof}
    Each of the $n$ partitions has two distinct $\attB$ values in $\Rend$,
    neither of which is $0$ as in $\Rbeg$. Thus, two assignment modifiers
    are required. By Lemma~\ref{barrier}, a range cannot go across barriers,
    thus the total range count includes at least two ranges per partition.
\end{proof}
\begin{lemma}
    $S$ has a common supersequence of length at most $t$ iff the
    best diff between $\Rbeg$ and $\Rend$ under $\famop{\rangeunion}{\gets}$
    has cost at most $t+2n(t+99)$.
\end{lemma}
\begin{proof}
    ($\Leftarrow$) Let $F=(f_1,\ldots,f_m)$ be a best diff
    between $\Rbeg$ and $\Rend$ under $\famop{\rangeunion}{\gets}$,
    where \[f_i=(\attA\in\bigcup_{j=1}^{r_i}[a_{ij},z_{ij}],\attB\gets b_i)\]
    with cost at most $t+2n(t+99)$. The total range count of $F$
    cannot exceed $2n$, otherwise its cost must be at least
    $(t+99)(2n+1)>t+2n(t+99)$. Together with Lemma~\ref{rangecountlowerbound},
    the total range count of $F$ must be exactly $2n$.
    The cost implies that $m\leq t$.

    Let $s=b_1\ldots b_m$. Because each partition contains two distinct
    $\attB$ values in $\Rend$, and because the total range count must be $2n$,
    there must be exactly two operations that affects
    each partition. For $k\in[n]$, partition $P_k$ has two operations $f_i$ and $f_j$,
    where $i<j$, that affects it. The operations must be such that $b_i=u_k$ and $b_j=v_k$,
    where $f_i$ sets the $\attB$ value for all tuples in $P_k$ to $u_k$,
    and $f_j$ then sets the $\attB$ value for one tuple to $v_k$.
    Hence, removing symbols from $s$ except at indices $i$ and $j$ would yield
    $b_ib_j=u_kv_k=s_k$.
    Thus, $s$ is a supersequence of $S$ with length $m\leq t$.

    ($\Rightarrow$) Let $s=w_1\ldots w_m$ be a common supersequence of $S$
    of length $m\leq t$. For each symbol $c$, let $\mathrm{first}(c)$ be the
    smallest $i$ such that $w_i=c$; and $\mathrm{last}(c)$, largest.

    Construct $F=(f_1,\ldots,f_m)$ where
    \begin{align*}
    R^{(1)}_i &= \bigcup_{\substack{k\in[n]\\ \mathrm{first}(u_k)=i}} [4k-3,4k-1]\\
    R^{(2)}_i &= \bigcup_{\substack{k\in[n]\\ \mathrm{last}(v_k)=i}} [4k-2,4k-2]\\
    f_i &= (\attA\in R^{(1)}_i\cup R^{(2)}_i,\attB\gets w_i)
    \end{align*}

    For $k\in[n]$, partition $P_k$ has two operations $f_i$ and $f_j$
    that affects it, where $i=\mathrm{first}(u_k)$ and $j=\mathrm{last}(v_k)$.
    Because $s$ is a supersequence of $S$, we have $i=\mathrm{first}(u_k)<\mathrm{last}(v_k)=j$.
    Hence, $f_i$ sets the $\attB$ value for all tuples in $P_k$ to $u_k$,
    and $f_j$ then sets the $\attB$ value for one tuple to $v_k$.
    The total range count of $F$ is $2n$.
    Thus, $F$ is a diff
    between $\Rbeg$ and $\Rend$ under $\famop{\rangeunion}{\gets}$
    whose cost is $m+2n(t+99)\leq t+2n(t+99)$.
\end{proof}

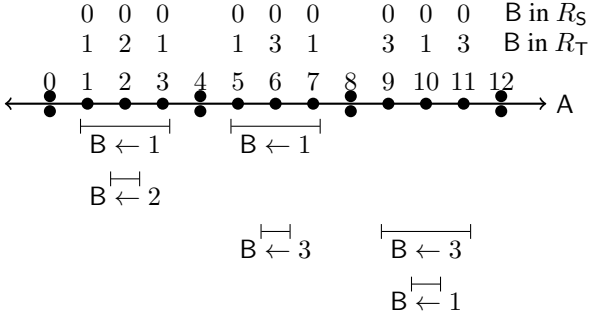
\begin{figure}[!hbt]
\centering
\vspace{-10pt}
\begin{tikzpicture}
    \draw[<->,thick] (-0.1,0)--(7.1,0) node[anchor=west] {$\attA$};
    \foreach \i/\j in {1/1,2/2,3/1,5/1,6/3,7/1,9/3,10/1,11/3} {
        \node () at (0.5+\i*0.5,1.2) {$0$};
        \node () at (0.5+\i*0.5,0.8) {$\j$};
        \node () at (0.5+\i*0.5,0.3) {$\i$};
        \fill (0.5+\i*0.5,0.0) circle (0.08cm);
    }
    \foreach \i in {0,4,8,12} {
        \node () at (0.5+\i*0.5,0.3) {$\i$};
        \fill (0.5+\i*0.5,0.1) circle (0.08cm);
        \fill (0.5+\i*0.5,-0.1) circle (0.08cm);
    }
    \node () at (7.1,1.2) {$\attB$ in $\Rbeg$};
    \node () at (7.1,0.8) {$\attB$ in $\Rend$};
    \draw[|-|] (0.9,-0.3)--node[below=-0.5pt]{$\attB\gets 1$} (2.1,-0.3);
    \draw[|-|] (2.9,-0.3)--node[below=-0.5pt]{$\attB\gets 1$} (4.1,-0.3);
    \draw[|-|] (1.3,-1.0)--node[below=-0.5pt]{$\attB\gets 2$} (1.7,-1.0);
    \draw[|-|] (3.3,-1.7)--node[below=-0.5pt]{$\attB\gets 3$} (3.7,-1.7);
    \draw[|-|] (4.9,-1.7)--node[below=-0.5pt]{$\attB\gets 3$} (6.1,-1.7);
    \draw[|-|] (5.3,-2.4)--node[below=-0.5pt]{$\attB\gets 1$} (5.7,-2.4);
\end{tikzpicture}
\caption{Illustration of Example \ref{scsexample}}
\vspace{-10pt}
\label{fig:scsillus}
\end{figure}

\begin{example}
\label{scsexample}
Consider the \textsc{2DistinctSCS} instance with $S=\{12,13,31\}$ and $t=4$.
The string $1231$ is a supersequence of $S$. The reduction
gives the instance of the
$\bdI(\famop{\rangeunion}{\gets})$ problem shown in Figure~\ref{fig:scsillus}.

In this case, $F=(f_1,f_2,f_3,f_4)$ where
\begin{align*}
f_1 &= (\attA\in [1,3]\cup[5,7],\attB\gets 1)\\
f_2 &= (\attA\in [2,2],\attB\gets 2)\\
f_3 &= (\attA\in [6,6]\cup[9,11],\attB\gets 3)\\
f_4 &= (\attA\in [10,10],\attB\gets 1)
\end{align*}
is a best diff with cost $4+2\cdot 13\cdot 3=82$. Dropping $f_4$ and
changing the condition of $f_1$ to $\attA\in [1,3]\cup[5,7]\cup[10,10]$ yield
lower cost but is not permissible, because the tuple at $\attA=10$
would have an incorrect $\attB$ value.
\end{example}

\begin{theorem}
    The $\bdI(\famop{\rangeunion}{+})$,
    $\bdI(\famop{\rangeunion}{\getsadd})$, and
    $\bdI(\famop{\rangeunion}{\affine})$ problems
    are $\NP$-hard.
\label{unionadd}
\label{uniongetsadd}
\label{unionaffine}
\end{theorem}
\begin{proof}
This follows from polynomial-time reductions from respective range versions,
using the same instance and setting $\kappa_0=0$ and $\kappa_1=1$.
\end{proof}

%% file: solution-aab-compose.tex
\section{Relaxation: $\bdII$ problems}
\label{sec:solution-aab-compose}

In this section, we discuss a relaxation to the constraints
of the ``base case'' in the previous section. We allow the
number of read-only attributes to be more than one.
We see in the previous section that even when restricted to
$1$ read-only attribute, the problem becomes $\NP$-hard even
with relatively simple conditions and modifiers. With
more attributes, the problem reaches the
hardness boundary much more quickly.

The $\bdII(\calF)$ problem is similar to the best diff $\bd(\calF)$
problem, but constrained to one write-only attribute and
no read-write attributes;
the number of read-only attributes may vary.
Let $\calA=\{\attA_1,\ldots,\attA_\mu\}$ and $\calB=\{\attB\}$,
where $\attB\not\in\calA$.

\subsection{With Equality Conditions}
Unlike in the previous section with $1$ read-only attribute,
the problem becomes $\NP$-hard even with \emph{equality} conditions.
For the \emph{assignment} case, the problem is closely related to the view synthesis problem,
and we derive the hardness result through it. For the remaining cases,
we show hardness through reductions from the $1$ read-only attribute version
with range conditions.

\begin{theorem}
    The $\bdII(\famop{=}{\gets})$ problem
    is $\NP$-hard.
\label{multeqgets}
\end{theorem}

The proof of this theorem is given in Appendix~\ref{sec:additionalproofs}.

\begin{theorem}
    The $\bdII(\famop{=}{+})$,
    $\bdII(\famop{=}{\getsadd})$, and
    $\bdII(\famop{=}{\affine})$ problems
    are $\NP$-hard.
\label{multeqrest}
\end{theorem}
We only show the proof for $\bdII(\famop{=}{+})$,
as the remaining proofs are similar.
The idea is to simulate range conditions in one attribute
with equality conditions in multiple attributes.

Consider the instance $(\attK,\calA,\calB,\Rbeg,\Rend)$
of $\bdI(\famop{\range}{+})$, where $\calA=\{\attA\}$,
and $a_1<\cdots<a_\ell$ are values in $V_\attA(\Rbeg,\Rend)$ in order.
The reduction is as follows: we construct the instance
$(\attK,\calA',\calB,\Rbeg',\Rend')$.
Here, $\calA'=\{\attA_1,\ldots,\attA_{2\ell}\}$, and
$\Rbeg'$ and $\Rend'$ are identical to $\Rbeg$ and $\Rend$, respectively,
except that the attribute $\attA=a$ is replaced by $\attA_1,\ldots,\attA_{2\ell}$
in the following fashion.

Let $a^*=a_\ell+1\not\in V_\attA(\Rbeg,\Rend)$. Define
$\rho_i=[a_i,a_\ell]$ and $\rho_{\ell+i}=[a_1,a_i]$ for $i\in[\ell]$.
For each tuple, its $\attA_i$ value is assigned as its $\attA$ value,
except when the value is in the range $\rho_i$, in which case it is
assigned to $a^*$, for $i\in[2\ell]$.

\begin{lemma}
    Best diffs in $\diff(\Rbeg,\Rend,\famop{\range}{+})$
    have cost $n$ (in $\bdI$) if and only if best diffs in
    $\diff(\Rbeg',\Rend',\famop{\leq}{+})$
    have cost $n$ (in $\bdII$).
\label{rangeaddandmulteqadd}
\end{lemma}
\begin{proof}
    The following pair of conditions are equivalent,
    matching corresponding tuples in the respective problems.
    \begin{align*}
      \attA\in[a_i,a_j] &\equiv \attA_i=a^*\wedge \attA_{\ell+j}=a^*\\
      \attA\in[a_1,a_j] &\equiv \attA_{\ell+j}=a^*\\
      \attA\in[a_i,a_\ell] &\equiv \attA_i=a^*
    \end{align*}
    where $i,j\in[\ell]$, and
    \begin{align*}
      \attA\in[a_j,a_j] &\equiv \attA_i=a_j
    \end{align*}
    where $i,j\in[2\ell]$.
    Thus, the best diffs translate from one problem to the other.
\end{proof}

\subsection{With At-most Conditions}
The classification for the problem with the \emph{assignment} modifier is unknown.
The cases with the \emph{increment}, \emph{assignment/increment}, and
\emph{affine} modifiers are $\NP$-hard, even if we
restrict the number of read-only attributes to $2$.

\begin{theorem}
    The $\bdII(\famop{\leq}{+})$,
    $\bdII(\famop{\leq}{\getsadd})$, and
    $\bdII(\famop{\leq}{\affine})$ problems
    are $\NP$-hard, even with
    $2$ read-only attributes.
\label{multleq}
\end{theorem}
We only show the proof for $\bdII(\famop{\leq}{+})$,
as the remaining proofs are similar.
The idea is to simulate range conditions in one attribute
with at-most conditions in two attributes.

Consider the instance $(\attK,\calA,\calB,\Rbeg,\Rend)$
of $\bdI(\famop{\range}{+})$, where $\calA=\{\attA\}$.
The reduction is as follows: we construct the instance
$(\attK,\calA',\calB,\Rbeg',\Rend')$.
Here, $\calA'=\{\attA_1,\attA_2\}$, and
$\Rbeg'$ and $\Rend'$ are identical to $\Rbeg$ and $\Rend$, respectively,
except that the attribute $\attA=a$ is replaced by $\attA_1=a$ and $\attA_2=-a$.

\begin{lemma}
    Best diffs in $\diff(\Rbeg,\Rend,\famop{\range}{+})$
    have cost $n$ (in $\bdI$) if and only if best diffs in
    $\diff(\Rbeg',\Rend',\famop{\leq}{+})$
    have cost $n$ (in $\bdII$).
\label{rangeaddandmultleqadd}
\end{lemma}
\begin{proof}
    The following pair of conditions are equivalent,
    matching corresponding tuples in the respective problems.
    \begin{align*}
      \attA\in[a,z] &\equiv \attA_1\leq z\wedge \attA_2\leq -a\\
      \attA\in[v_\textrm{min}^\attA,z] &\equiv \attA_1\leq z\\
      \attA\in[a,v_\textrm{max}^\attA] &\equiv \attA_2\leq -a
    \end{align*}
    Thus, the best diffs translate from one problem to the other.
\end{proof}

\subsection{With At-most/At-least, Range, or Union-of-Ranges Conditions}
With \emph{at-most/at-least}, \emph{range}, or \emph{union-of-ranges} conditions,
the problem is $\NP$-hard
when using \emph{increment}, \emph{assignment/increment}, or
\emph{affine} modifiers, via trivial reductions.
\begin{theorem}
    The $\bdII(\famop{\phi}{\omega})$ problem,
    for
    \begin{align*}
        \phi &\in\{\lgeq,\range,\rangeunion\}\text{ and}\\
        \omega &\in\{+,\getsadd,\affine\}
    \end{align*}
    is $\NP$-hard, even with
    $1$ read-only attribute.
\label{multtrivial}
\end{theorem}
\begin{proof}
    These are generalizations from their $\bdI$ counterparts,
    which are all $\NP$-hard.
\end{proof}

The case with the \emph{assignment} modifier is different.
The classification for the problem with the \emph{at-most/at-least} condition is unknown.
While there is a version of the view synthesis problem that is similar
to the \emph{range} case, we provide a proof of $\NP$-hardness via a
different problem. The proof works even when we
restrict the number of read-only attributes to $2$.
\begin{theorem}
    The $\bdII(\famop{\range}{\gets})$ problem
    is $\NP$-hard, even with
    $2$ read-only attributes.
\label{multrangegets}
\end{theorem}

In order to prove Theorem~\ref{multrangegets},
we provide a polynomial-time reduction from \textsc{RectangleCover},
which is a known $\NP$-hard problem, defined as follows~\cite{Culberson19942}.
\begin{definition}[RectangleCover]
    The \textsc{RectangleCover} decision problem is, given
    an orthogonal polygon $P$ (on a plane) with $n$ vertices,
    and a nonnegative integer $t$, determine whether
    there is a rectangle cover of $P$ of size $t$; that is,
    whether there exists a set of $t$ axis-aligned rectangles
    whose union is exactly $P$.
\end{definition}

\begin{proof}
Consider an instance of the \textsc{RectangleCover} problem with an
orthogonal polygon $P$ with $n$ vertices. Without loss of generality,
let $[\ell]$ be the set of coordinates used by $P$, where $\ell\leq n$.
(Essentially we perform a ``rank-space reduction''~\cite{alstrup2000new}, since stretching
the polygon does not affect the size of the cover.)
Construct an $\ell\times\ell$ grid and superimpose the polygon $P$ on it.

We create $99$ tuples for each of the $\ell\times\ell$ grid cells,
with their $\attA_1$ and $\attA_2$ values corresponding to their $x$ and $y$
coordinates. Their $\attB$ values are set as follows:
in $\Rbeg$, set all $\attB$ values to distinct positive values;
in $\Rend$, set $\attB$ to the same values as in $\Rbeg$, except when the
following condition applies: for the tuple with $\attA_1=x$ and $\attA_2=y$,
the square with opposite corners $(x,y)$ and $(x+1,y+1)$ is contained in
(the superimposed) $P$;
in which case $\attB$ is set to $0$.

The claim is that $P$ has a rectangle cover of size $t$ if and only if
$\Rbeg$ and $\Rend$ has a diff under $\famop{\range}{\gets}$ of cost $t$.
This is because using a rectangle with opposite corners $(x_1,y_1)$
and $(x_2,y_2)$, where $x_1 < x_2$ and $y_1 < y_2$,
corresponds to setting $\attB\gets 0$ to the
tuples matching the condition
$\attA_1\in[x_1,x_2-1]\wedge\attA_2\in[y_1,y_2-1]$.
To see why the off-by-one correction is needed,
consider the case $x_1=x_2$. The range $[x_1,x_2]$ is not
empty (contains one element), but the rectangle defined
by those $x$-coordinates have zero width.

The reason we need multiple tuples per grid cell is
to prevent ``unsetting'' the $\attB$ value from $0$.
Once the $\attB$ values of these
$99$ tuples are set to $0$ (or any value),
they cannot be changed back into distinct $\attB$ values again
using assignment modifier.

Thus, the problem is $\NP$-hard.
\end{proof}

\begin{theorem}
    The $\bdII(\famop{\rangeunion}{\gets})$ problem
    is $\NP$-hard, even with
    $2$ read-only attributes.
\label{multuniongets}
\end{theorem}
\begin{proof}
    This is a corollary of Theorem~\ref{multrangegets}.
\end{proof}

%% file: futurework.tex
\section{Conclusions and Future Work}
\label{sec:future}

This paper introduces the family of \ddiff problems
characterized by a particular set
of modifiers and conditions of interest. It identifies
the base case of $1$ read-only and $1$ write-only attribute
and fully classifies the complexity across families of
operations (Table~\ref{table:summary1}).
It also discusses the generalization to
multiple read-only attributes, showing $\NP$-hardness
in most families of operations (Table~\ref{table:summary2}).

Some remaining open problems are discussed earlier,
particularly characterizing $\bdII(\famop{\leq}{\gets})$
and $\bdII(\famop{\lgeq}{\gets})$. In addition, we have
only discussed the settings with $1$ write-only attribute
and $0$ read-write attributes. In particular, introducing
read-write attributes creates a complexity where an
operation may modify the values in the attributes used
for conditions, and therefore the same condition may match
a different set of tuples depending on when it is used,
making the order of operations even more crucial.
Characterizing the problem under relaxations of these constraints
is therefore an interesting venue for further investigation.

\section*{Acknowledgements}
We would like to thank Jeff Erickson for initial discussions;
we would also like to thank Liqi Xu and Sheng Shen for practical implementations of \ddiff.

%% file: appendix.tex
\section{Polynomial-time Results}
\label{sec:polytimeproofs}
In this section, we discuss cases of the \ddiff problem
with polynomial time algorithms in more detail.

\begin{proof}[of Theorem~\ref{eqgets} ($\famop{=}{\gets}$, $\famop{=}{+}$, $\famop{=}{\getsadd}$, $\famop{=}{\affine}$)]
    Each operation affects the value of all tuples with one value of $\attA$,
    and there is not a reason to apply two operations for the same value of
    $\attA$. Therefore, one simply needs to sort the tuples by their $\attA$
    value, iterate through all values of $\attA$,
    and create an appropriate operation to modify the $\attB$ value
    to the right value, if possible.
\end{proof}

\begin{proof}[of Theorem~\ref{leqgets} ($\famop{\leq}{\gets}$)]
    We first consider the following proposition:
    if $\diff(\Rbeg,\Rend,\famop{\leq}{\gets})$
    is nonempty, then it
    contains a best diff $F=(f_1,\ldots,f_m)$
    in which for all $i,j\in[m]$, if $i<j$ and
    \begin{align*}
        f_i&=(\attA\leq a_i,\attB\gets b_i)\text{ and}\\
        f_j&=(\attA\leq a_j,\attB\gets b_j)
    \end{align*}
    then $a_i > a_j$.

    The proof follows.
    Suppose $a_i \leq a_j$, then the diff $F'$ constructed by
    removing $f_i$ from $F$ achieves the same result---that is,
    $F'(\Rbeg)=F(\Rbeg)=\Rend$---because changes caused by $f_i$
    are rendered moot by $f_j$,
    and thus $F$ is not a best diff.
    \begin{center}
    \begin{tikzpicture}
        \begin{scope}[xshift=0.0cm]
            \draw[thin] (-0.2,-1.5) rectangle (3.7,0.3);
            \draw[<->,thick] (-0.1,0)--(3.1,0) node[anchor=west] {$\attA$};
            \draw[<-|] (0,-0.3)--node[below=-1.5pt]{$\attB\gets b_i$} (2.0,-0.3) node[anchor=west] {$a_i$};
            \draw[<-|] (0,-1.0)--node[below=-1.5pt]{$\attB\gets b_j$} (2.7,-1.0) node[anchor=west] {$a_j$};
        \end{scope}
        \node () at (4.0,-0.6) {$=$};
        \begin{scope}[xshift=4.5cm]
            \draw[thin] (-0.2,-1.5) rectangle (3.7,0.3);
            \draw[<->,thick] (-0.1,0)--(3.1,0) node[anchor=west] {$\attA$};
            \draw[<-|] (0,-1.0)--node[below=-1.5pt]{$\attB\gets b_j$} (2.7,-1.0) node[anchor=west] {$a_j$};
        \end{scope}
    \end{tikzpicture}
    \end{center}
    Therefore, one simply needs to sort the tuples
    by their $\attA$, and in decreasing order of $\attA$,
    create an appropriate operation to modify the $\attB$ value
    to the right value, if possible.
\end{proof}

\begin{proof}[of Theorem~\ref{leqadd} ($\famop{\leq}{+}$)]
    We first consider the following proposition:
    if $\diff(\Rbeg,\Rend,\famop{\leq}{+})$
    is nonempty, then it
    contains a best diff $F=(f_1,\ldots,f_m)$
    in which for all $i,j\in[m]$, if $i<j$ and
    \begin{align*}
        f_i&=(\attA\leq a_i,\attB\gets\attB+b_i)\text{ and}\\
        f_j&=(\attA\leq a_j,\attB\gets\attB+b_j)
    \end{align*}
    then $a_i > a_j$. The correctness of the proposition follows from the fact
    that the operations are commutative.

    Therefore, one simply needs to sort the tuples
    in decreasing order of $\attA$,
    create an appropriate operation to modify the $\attB$ value
    to the right value, if possible.
\end{proof}

\begin{proof}[of Theorem~\ref{leqgetsadd} ($\famop{\leq}{\getsadd}$)]
    We first consider the following proposition:
    if $\diff(\Rbeg,\Rend,\famop{\leq}{\getsadd})$
    is nonempty, then it
    contains a best diff $F=(f_1,\ldots,f_m)$
    in which for all $i,j\in[m]$, if $i<j$ and one of the following is true:
    \begin{center}
    {\setlength{\tabcolsep}{2pt}
    \begin{tabular}{l l l}
    (a)&$f_i=(\attA\leq a_i,\attB\gets b_i)$,&
    $f_j=(\attA\leq a_j,\attB\gets b_j)$\\
    (b)&$f_i=(\attA\leq a_i,\attB\gets\attB+b_i)$,&
    $f_j=(\attA\leq a_j,\attB\gets b_j)$\\
    (c)&$f_i=(\attA\leq a_i,\attB\gets\attB+b_i)$,&
    $f_j=(\attA\leq a_j,\attB\gets\attB+b_j)$\\
    (d)&$f_i=(\attA\leq a_i,\attB\gets b_i)$,&
    $f_j=(\attA\leq a_j,\attB\gets\attB+b_j)$
    \end{tabular}}
    \end{center}
    then $a_i > a_j$. Equivalently, we define an \emph{inversion} as
    a pair $(i,j)$ such that the preconditions hold but instead
    $a_i\leq a_j$, and we claim that there exists a best diff without inversions.

    The proof follows.
    For cases (a) and (b), the proof is the same as in Theorem~\ref{leqgets}.
    For cases (c) and (d), the proof is as follows.
    Here, an inversion, as defined above, is a pair
    $(i,j)\in[m]$ such that $i<j$,
    the condition for $f_i$ is $\attA\leq a_i$,
    the condition for $f_j$ is $\attA\leq a_j$,
    and $a_i > a_j$.

    If $\diff(\Rbeg,\Rend,\famop{\leq}{+})$
    is nonempty, then it
    contains a best diff $F=(f_1,\ldots,f_m)$
    that does not violate (a), or (b),
    with the fewest inversions.
    We show that $F$ has zero inversions.

    Suppose $F$ has an inversion of type (c) or type (d), that is,
    there are $i,j\in[m]$ such that $i<j$ and either
    \begin{align*}
        f_i&=(\attA\leq a_i,\attB\gets\attB+b_i)\text{ and}\\
        f_j&=(\attA\leq a_j,\attB\gets\attB+b_j)
    \end{align*}
    (for an inversion of type (c))
    \begin{align*}
        f_i&=(\attA\leq a_i,\attB\gets b_i)\text{ and}\\
        f_j&=(\attA\leq a_j,\attB\gets\attB+b_j)
    \end{align*}
    (for an inversion of type (d)), and $a_i \leq a_j$. Without loss of generality, let $(i,j)$
    be such a pair where the index difference $j-i$ is the smallest.
    It must be the case that $j-i=1$, for otherwise $f_k$
    where $i<k<j$ will create a violation for (a) or (b).

    If $(f_i, f_j)$ constitutes an inversion of type (c), we may create
    a new diff $F'$ is equivalent to $F$, but with $f_i$ and $f_j$ switched.
    $F'(\Rbeg)=F(\Rbeg)=\Rend$ since the increment updates are commutative
    (and there are no operations $f_k$, $i<k<j$ with assignment updates to break
    said commutativity since $j-i=1$).

    Otherwise, $(f_i, f_j)$ constitutes an inversion of type (d).
    In this case, let $F'=(f_1,\ldots,f_{i-1},f_j,g,f_{j+1},\ldots,f_m)$
    where
    \begin{align*}
        g&=(\attA\leq a_i,\attB\gets b_i+b_j)
    \end{align*}
    then $F'(\Rbeg)=F(\Rbeg)=\Rend$. In either case (inversion of type (c)
    or of type (d)), $F'$ has one fewer inversion than $F$,
    a contradiction to the fact that $F$ has the fewest inversions.
    \begin{center}
    \begin{tikzpicture}
        \begin{scope}[xshift=0.0cm]
            \draw[thin] (-0.2,-1.5) rectangle (3.7,0.3);
            \draw[<->,thick] (-0.1,0)--(3.1,0) node[anchor=west] {$\attA$};
            \draw[<-|] (0,-0.3)--node[below=-1.5pt]{$\attB\gets b_i$} (2.0,-0.3) node[anchor=west] {$a_i$};
            \draw[<-|] (0,-1.0)--node[below=-1.5pt]{$\attB\gets\attB+b_j$} (2.7,-1.0) node[anchor=west] {$a_j$};
        \end{scope}
        \node () at (4.0,-0.6) {$=$};
        \begin{scope}[xshift=4.5cm]
            \draw[thin] (-0.2,-1.5) rectangle (3.7,0.3);
            \draw[<->,thick] (-0.1,0)--(3.1,0) node[anchor=west] {$\attA$};
            \draw[<-|] (0,-0.3)--node[below=-1.5pt]{$\attB\gets\attB+b_j$} (2.7,-0.3) node[anchor=west] {$a_j$};
            \draw[<-|] (0,-1.0)--node[below=-1.5pt]{$\attB\gets b_i+b_j$} (2.0,-1.0) node[anchor=west] {$a_i$};
        \end{scope}
    \end{tikzpicture}
    \end{center}
    Therefore, one simply needs to sort the tuples
    by their $\attA$, and compute the smallest number of
    operations required using dynamic programming.
\end{proof}

\begin{proof}[of Theorem~\ref{leqaffine} ($\famop{\leq}{\affine}$)]
    We first consider the following proposition:
    if $\diff(\Rbeg,\Rend,\famop{\leq}{\affine})$
    is nonempty, then it
    contains a best diff $F=(f_1,\ldots,f_m)$
    in which for all $i,j\in[m]$, if $i<j$ and
    \begin{align*}
        f_i&=(\attA\leq a_i,\attB\gets b_i\attB+c_i)\text{ and}\\
        f_j&=(\attA\leq a_j,\attB\gets b_j\attB+c_j)
    \end{align*}
    then $a_i > a_j$. Equivalently, we define an \emph{inversion} as
    a pair $(i,j)$ such that the preconditions hold but instead
    $a_i\leq a_j$, and we claim that there exists a best diff without inversions.

    The proof follows. Suppose $F$ has an inversion, that is,
    there are $i,j\in[m]$ such that $i<j$ and
    \begin{align*}
        f_i&=(\attA\leq a_i,\attB\gets b_i\attB+c_i)\text{ and}\\
        f_j&=(\attA\leq a_j,\attB\gets b_j\attB+c_j)
    \end{align*}
    and $a_i \leq a_j$. Without loss of generality, let $(i,j)$
    be such a pair where the index difference $j-i$ is the smallest.
    It must be the case that $j-i=1$, for otherwise $f_k$
    where $i<k<j$ will be such that either $(i,k)$ or $(k,j)$
    is an inversion with a smaller index difference.

    Let $F'=(f_1,\ldots,f_{i-1},f_j,g,f_{j+1},\ldots,f_m)$
    where
    \begin{align*}
        g&=(\attA\leq a_i,\attB\gets b_i\attB+c')\text{ and}\\
        c' &=c_j+b_jc_i-c_jb_i
    \end{align*}
    then $F'(\Rbeg)=F(\Rbeg)=\Rend$, but $F'$ has one fewer inversion than $F$,
    a contradiction to the fact that $F$ has the fewest inversions.
    \begin{center}
    \begin{tikzpicture}
        \begin{scope}[xshift=0.0cm]
            \draw[thin] (-0.2,-1.5) rectangle (3.7,0.3);
            \draw[<->,thick] (-0.1,0)--(3.1,0) node[anchor=west] {$\attA$};
            \draw[<-|] (0,-0.3)--node[below=-1.5pt]{$\attB\gets b_i\attB+c_i$} (2.0,-0.3) node[anchor=west] {$a_i$};
            \draw[<-|] (0,-1.0)--node[below=-1.5pt]{$\attB\gets b_j\attB+c_j$} (2.7,-1.0) node[anchor=west] {$a_j$};
        \end{scope}
        \node () at (4.0,-0.6) {$=$};
        \begin{scope}[xshift=4.5cm]
            \draw[thin] (-0.2,-1.5) rectangle (3.7,0.3);
            \draw[<->,thick] (-0.1,0)--(3.1,0) node[anchor=west] {$\attA$};
            \draw[<-|] (0,-0.3)--node[below=-1.5pt]{$\attB\gets b_j\attB+c_j$} (2.7,-0.3) node[anchor=west] {$a_j$};
            \draw[<-|] (0,-1.0)--node[below=-1.5pt]{$\attB\gets b_i\attB+c'$} (2.0,-1.0) node[anchor=west] {$a_i$};
        \end{scope}
    \end{tikzpicture}
    \end{center}
    Label the tuples $1,\ldots,n$ in order of increasing $\attA$.
    We can use dynamic programming to compute $f(m)$, the cost
    of modifying tuples $1$ throught $m$ in order to match $\Rend$.
    The process effectively segments the tuples into blocks, each
    of which containing tuples that can be transformed together using one affine transformation,
    i.e., lie on the same ``line'', taking extra care of
    constant transformations---those with the modifier $\attB\gets b\attB+c$
    with $b=0$.

    The final answer $f(n)$ can be computed in $O(N\log N)$.
\end{proof}

\begin{proof}[of Theorem~\ref{lgeqgets} ($\famop{\lgeq}{\gets}$)]
    We first consider the following proposition:
    if $\diff(\Rbeg,\Rend,\famop{\lgeq}{\gets})$
    is nonempty, then it
    contains a best diff $F=(f_1,\ldots,f_n)$
    in which there are no operations
    \begin{align*}
        f_i&=(\attA\leq a_i, \attB\gets b_i)\text{ and}\\
        f_j&=(\attA\geq a_j, \attB\gets b_j)
    \end{align*}
    such that $a_i\geq a_j$.
    In other words, $F$ contains no two ``overlapping'' operations.

    The proof follows.
    Let $F=(f_1,\ldots,f_m)$ be a best diff in
    $\diff(\Rbeg,\Rend,\famop{\lgeq}{\gets})$
    with the smallest total length.
    Assume to the contrary that there exists operations $f_i$ and $f_j$
    such that
    \begin{align*}
        f_i&=(\attA\leq a_i, \attB\gets b_i)\text{ and}\\
        f_j&=(\attA\geq a_j, \attB\gets b_j)
    \end{align*}
    and $a_i\geq a_j$.
    If $i<j$, then let $g=(\attA\leq a_j-1, \attB\gets b_i)$,
    and $F'=(f_1,\ldots,f_{i-1},g,f_{i+1},\ldots,f_m)$. Then,
    $F'(\Rbeg)=F(\Rbeg)=\Rend$, but $F'$ has smaller total length than $F$.
    \begin{center}
    \begin{tikzpicture}
        \begin{scope}[xshift=0.0cm]
            \draw[thin] (-0.2,-1.5) rectangle (3.7,0.3);
            \draw[<->,thick] (-0.1,0)--(3.1,0) node[anchor=west] {$\attA$};
            \draw[<-|] (0,-0.3)--node[below=-1.5pt]{$\attB\gets b_i$} (1.7,-0.3) node[anchor=west] {$a_i$};
            \draw[|->] (1.3,-1.0) node[anchor=east] {$a_j$} --node[below=-1.5pt]{$\attB\gets b_j$} (3.0,-1.0);
        \end{scope}
        \node () at (4.0,-0.6) {$=$};
        \begin{scope}[xshift=4.5cm]
            \draw[thin] (-0.2,-1.5) rectangle (3.7,0.3);
            \draw[<->,thick] (-0.1,0)--(3.1,0) node[anchor=west] {$\attA$};
            \draw[<-|] (0,-0.3)--node[below=-1.5pt]{$\attB\gets b_i$} (1.2,-0.3) node[anchor=west] {$a_j-1$};
            \draw[|->] (1.3,-1.0) node[anchor=east] {$a_j$} --node[below=-1.5pt]{$\attB\gets b_j$} (3.0,-1.0);
        \end{scope}
    \end{tikzpicture}
    \end{center}
    Otherwise, if $i>j$, then let $g=(\attA\geq a_i+1, \attB\gets b_j)$,
    and $F'=(f_1,\ldots,f_{j-1},g,f_{j+1},\ldots,f_m)$. Then,
    $F'(\Rbeg)=F(\Rbeg)=\Rend$, but $F'$ has smaller total length than $F$.
    In either case, it contradicts with the fact that
    $F$ has the smallest total length.

    Therefore, one simply needs to sort the tuples
    by their $\attA$,
    decide on the ``breakpoint'' that separates the
    $\attA\leq a$ conditions from the $\attA\geq a$ conditions, then
    use the algorithm similar to the one given in Theorem~\ref{leqgets}
    on each side.
\end{proof}

\begin{proof}[of Theorem~\ref{rangegets} ($\famop{\range}{\gets}$)]
    If $\diff(\Rbeg,\Rend,\famop{\range}{\gets})$
    is nonempty, then it
    contains a best diff $F=(f_1,\ldots,f_n)$
    in which for any two operations
    $f_i=(\attA\in [a_i,z_i], \attB\gets b_i)$ and
    $f_j=(\attA\in [a_j,z_j], \attB\gets b_j)$ where
    $i<j$, either $z_i<a_j$ or $z_j<a_i$ or $a_i\leq a_j\leq z_j\leq z_i$.
    The proof is similar to that given in Theorem~\ref{lgeqgets}.

    Label the tuples $1,\ldots,n$ in order of increasing $\attA$.
    We can use dynamic programming to compute $f(m_1, m_2, \delta)$,
    the cost of modifying tuples $m_1$ through $m_2$ in order to match $\Rend$
    where all tuples have $\attB$ values set to $\delta$, unless $\delta$
    is $\textsc{Null}$ in which case all tuples have the
    original $\attB$ values like in $\Rbeg$.

    The final answer $f(1, n, \textsc{Null})$ can be computed in $O(N^4)$.
\end{proof}

\section{Approximation Results}
\label{sec:approxproofs}
For some $\NP$-hard cases of the \ddiff problem, we are able
to provide polynomial-time approximation algorithms. In fact,
these algorithms are discussed earlier, since they provide
exact results for different condition and modifier settings.

\subsection{Approximation for $\bdI(\famop{\lgeq}{+})$}

\begin{theorem}
    For the $\bdI(\famop{\lgeq}{+})$ problem
    an additive $1$-approximation can be found in $O(N\log N)$ time.
\label{lgeqaddapprox}
\end{theorem}
\begin{proof}
    In short, we show that if we only use the one of the ``at most''
    and ``at least'' condition type exclusively, and the cost of the
    best diff only increases by at most one.

    Let $F=(f_1,\ldots,f_m)$ be the best diff between $\Rbeg$ and $\Rend$,
    such that $I_\leq$ is the set of indices $i\in[m]$
    where $f_i$ is in the form $(\attA\leq a_i,\attB\gets\attB+b_i)$,
    and $I_\geq$ is the set of indices $i\in[m]$
    where $f_i$ is in the form $(\attA\geq a_i,\attB\gets\attB+b_i)$.
    Then, $F'=(f'_1,\ldots,f'_m,g)$ where, for $i\in[m]$,
    \begin{align*}
    f'_i&=\begin{cases}
    (\attA\leq a_i,\attB\gets\attB+b_i)&\text{if $i\in I_\leq$}\\
    (\attA\leq a_i-1,\attB\gets\attB-b_i)&\text{if $i\in I_\geq$}
    \end{cases}\\
    g&=(\attA\leq v_\attA^\mathrm{max},\attB\gets\attB+\sum_{i\in I_\geq} b_i)
    \end{align*}
    is also a diff between $\Rbeg$ and $\Rend$.
    By Theorem~\ref{leqadd}, a diff better or as good as $F'$ can be found in $O(N\log N)$.
\end{proof}

\subsection{Approximation for $\bdI(\famop{\range}{+})$}

\begin{theorem}
    For the $\bdI(\famop{\range}{+})$ problem,
    a multiplicative $2$-approximation can be found in $O(N\log N)$ time.
\label{rangeaddapprox}
\end{theorem}
\begin{proof}
    Let $F=(f_1,\ldots,f_m)$ be the best diff between $\Rbeg$ and $\Rend$,
    where $f_i=(\attA\in[a_i,z_i],\attB\gets\attB+b_i)$
    for $i\in[m]$.
    Then, $F'=(f'_1,\ldots,f'_{2m})$ where, for $i\in[m]$,
    \begin{align*}
    f'_{2i-1}&=(\attA\leq a_i-1,\attB\gets\attB-b_i)\\
    f'_{2i}&=(\attA\leq z_i,\attB\gets\attB+b_i)
    \end{align*}
    is also a diff between $\Rbeg$ and $\Rend$.
    By Theorem~\ref{leqadd}, a diff better or as good as $F'$ can be found in $O(N\log N)$.
\end{proof}

As a side note, there is the following reduction to the
edge-cost flow problem. Construct a flow network where
each vertex corresponds to a tuple. Let $v_1,\ldots,v_n$
be vertices corresponding to tuples in increasing order of $\attA$.
Construct an edge between every pair of vertices.
Assume $\attB$ of a vertex changes from $b$ to $b'$.
If $b<b'$, then let the supply of that vertex be $b'-b$.
If $b>b'$, then let the demand of that vertex be $b-b'$.

The claim is that there is that there is a diff of cost $m$
if and only if there is a flow of cost $m$. Thus, any approximation
scheme for edge-cost flow can also be used for $\bd$.

\section{The \textsc{2DistinctSCS} Problem}
\label{sec:2scs}
We prove that \textsc{2DistinctSCS} is
$\NP$-hard, via a polynomial-time reduction from \textsc{2SCS}.

Consider an instance of 2SCS with set $S$ of strings of length
two and a nonnegative integer $t$. Let $C$ be the set of symbols $c$
where $cc$ is in $S$.

For each symbol $c\in C$, create two new symbols $c_1$ and $c_2$.
Let $S'=\{f_1(u)f_2(v)\mid uv\in S\}$ where, for $i\in[2]$,
\[f_i(c)=\begin{cases}c_i&\text{if $c\in C$}\\c&\text{otherwise}\end{cases}\]

\begin{theorem}
$S$ has a supersequence of length at most $t$ iff
$S'$ has a supersequence of length at most $t$.
\end{theorem}
\begin{proof}
    ($\Leftarrow$) Let $s'$ be a supersequence of $S'$ of length
    at most $t$. Changing $c_1$ and $c_2$ to $c$ for $c\in C$
    from $s'$ and $S'$ preserves the supersequence constraint:
    characters at the same indices can be removed from $s'$ (or $s$) to
    obtain each string in $S'$ (or $S$).

    ($\Rightarrow$) Let $s=w_1\ldots w_m$ be a supersequence of $S$ of length
    $m\leq t$. Let $s'$ be identical to $s$, with the following change:
    for each $c\in C$, change its first occurrence in $s$ to $c_1$
    and its last occurrence in $s$ to $c_2$. By definition, each symbol
    in $c$ must have at least two occurrences in $s$, so there is no conflict.

    Consider $uv\in S$. Let $i,j\in[m]$ be indices such that $i<j$ and
    $w_i=u$ and $w_j=v$.
    The first occurence of $u$ in $s$ is at index
    $i'\leq i$ and the last occurence of $v$ in $s$ is at index
    $j'\geq j$. Thus, removing symbols other than at indices $i'$ and $j'$
    from $s'$ would give $w_{i'}w_{j'}=f_1(u)f_2(v)$.
    Therefore, $s'$ is a supersequence of $S'$ of length at most $t$.
\end{proof}

Therefore \textsc{2DistinctSCS} is $\NP$-hard.

\section{Additional Proofs}
\label{sec:additionalproofs}
\begin{proof}[of Theorem~\ref{lgeqandrange}]
    ($\Leftarrow$) By Lemma \ref{lgeqinteger},
    let $F=(f_1,\ldots,f_n)$ be a bounded best diff
    in $\diff(\Rbeg,\Rend,\famop{\lgeq}{+})$
    of cost $n$.
    Entries in the $\attA$ attribute in $\Rbeg$ and $\Rend$ are integers in $\{0,\ldots,n\}$.
    We can construct $F'$ from $F$ as follows.
    \begin{itemize}
    \item If $f_i=(\attA\leq a_i,\attB\gets\attB+b_i)$, then we construct
    $f'_i=(\attA\in [0,a_i],\attB\gets\attB+b_i)$,
    since $\attA\leq a_i$ if and only if $\attA\in [0,a_i]$.
    \item If $f_i=(\attA\geq a_i,\attB\gets\attB+b_i)$, then we construct
    $f'_i=(\attA\in [a_i,n],\attB\gets\attB+b_i)$,
    since $\attA\geq a_i$ if and only if $\attA\in [a_i,n]$.
    \end{itemize}
    Hence, $F'=(f'_1,\ldots,f'_n)$ is a best diff in
    $\diff(\Rbeg,\Rend,\famop{\range}{+})$.

    ($\Rightarrow$)
    By Lemma \ref{rangenocollision},
    let $F=(f_1,\ldots,f_n)$, where $f_i=(\attA\in[a_i,z_i],\attB\gets\attB+b_i)$
    for $i\in[n]$, be a bounded best diff
    in $\diff(\Rbeg,\Rend,\famop{\range}{+})$
    of cost $n$ such that $|C(F)|=0$.
    Consider a directed graph $G=(V,E)$ where $V=\{0,\ldots,n+1\}$
    and $E=\{(a_i, z_i+1)\mid i\in[m]\}$. That is, for each operation $f_i$,
    there is a corresponding edge $(a_i,z_i+1)$ in $E$.

    From the construction of $G$,
    vertex $0$ has in-degree $0$, and vertex $n+1$ has out-degree $0$.
    In addition, any vertex in $V$ cannot have in-degree greater than $1$.
    Assume the contrary: $\exists i,j$ s.t. $z_i+1=z_j+1=k$ then
    $z_i=z_j$.
    Likewise, any vertex in $V$ cannot have out-degree greater than $1$.
    A directed graph whose maximum in-degree and out-degree is $1$ can be
    decomposed into vertex-disjoint paths and cycles. However, $E$ only contains
    edges $(u,v)$ such that $u<v$, so $G$ does not contain cycles.
    Thus, $G$ can be decomposed into vertex-disjoint paths.

    Let $P=(v_1,\ldots,v_p)$ be a path in the vertex-disjoint path decomposition
    of $G$. We prove that $v_1=0$ or $v_p=n+1$.
    Assume for contradiction that $v_1\neq 0$ and $v_p\neq n+1$.
    From the degree requirements, we also have $v_1\neq n+1$ and $v_p\neq 0$.
    For each $k\in\{1,\ldots,p-1\}$,
    let $f^P_k=(\attA\in[v_k, v_{k+1}-1],\attB\gets\attB+b^P_k)$
    be an operation from $\{f_1,\ldots,f_n\}$ corresponding to the edge
    $(v_k,v_{k+1})$.

    Note that $f^P_{k-1}$ and $f^P_k$ are the only two operations
    in $\{f_1,\ldots,f_n\}$ that can affect the difference in the
    $\attB$ attribute between the tuples
    at $\attK=\attA=v_k-1$ and at $\attK=\attA=v_k$. In particular,
    it must be the case that $b^P_1=s_{v_1}>0$ and $-b^P_{p-1}=s_{v_p-1}>0$
    and $b^P_k-b^P_{k-1}=s_{v_k}>0$ in order for $F(\Rbeg)$ to agree with $\Rend$.
    However, these facts imply $0<b^P_1<b^P_2<\cdots <b^P_{p-1}$ and $b^P_{p-1}<0$,
    a contradiction. Therefore, $v_1=0$ or $v_p=n+1$.

    Define $b^P_0=b^P_p=0$.
    If $v_1=0$, we create the operation
    $f'_i=(\attA\leq v_{i+1}-1,\attB\gets\attB+(b^P_i-b^P_{i+1}))$
    for $i\in\{1,\ldots,p-1\}$.
    Otherwise, if $v_p=n+1$, we create the operation
    $f'_i=(\attA\geq v_i,\attB\gets\attB+(b^P_i-b^P_{i-1}))$
    for $i\in\{1,\ldots,p-1\}$.
    It follows that $F'_P=(f'_1,\ldots,f'_{p-1})$ are
    equivalent to $F_P=(f^P_1,\ldots,f^P_{p-1})$, and therefore
    $\diff(\Rbeg,\Rend,\famop{\lgeq}{+})$
    contains a bounded best diff of cost $n$.
\end{proof}

\begin{proof}[of Theorem~\ref{rangeaddandgetsadd}]
    ($\Leftarrow$) Any diff in $\diff(\Rbeg,\Rend,\famop{\range}{+})$
    is also a diff in $\diff(\Rbeg,\Rend,\famop{\range}{\getsadd})$.

    ($\Rightarrow$) Let $F=(f_1,\ldots,f_n)$ be a best diff in
    $\diff(\Rbeg,\Rend,\famop{\range}{\getsadd})$
    of cost $n$ that has the smallest number of assignment modifiers
    and, among the best diffs with the smallest number of assignment modifiers,
    has the smallest total length. We show that $F$ has no
    assignment modifiers.

    The proof follows.
    Assume to the contrary, and let $i$ be the smallest index in $[n]$ such that
    $f_i=(\attA\in[a_i,z_i],\attB\gets b_i)$ has an assignment modifier.

    Case 1: There is an operation $f_j=(\attA\in[a_j,z_j],\attB\gets\attB+b_j)$
    where $j<i$ and $a_j<a_i\leq z_j\leq z_i$. Then, let
    $f'_j=(\attA\in[a_j,a_i-1],\attB\gets\attB+b_j)$. If $F'$ is defined as
    $F$ where $f_j$ is replaced with $f'_j$, then $F'$ would still yield
    $F'(\Rbeg)=(\Rend)$, but the total length of $F'$ is smaller than that of $F$.

    Case 2: There is an operation $f_j=(\attA\in[a_j,z_j],\attB\gets\attB+b_j)$
    where $j<i$ and $a_i\leq a_j\leq z_i<z_j$. This case has an argument
    symmetric to Case 1.

    Case 3: There is an operation $f_j=(\attA\in[a_j,z_j],\attB\gets\attB+b_j)$
    where $j<i$ and $a_i\leq a_j<z_j\leq z_i$. If $F'$ is defined as
    $F$ where $f_j$ is removed, then $F'$ would still yield
    $F'(\Rbeg)=(\Rend)$, but the cost of $F'$ is smaller than that of $F$.

    Case 4: None of the above. Then, all tuples matching $\attA\in[a_i,z_i]$
    still has the same value in the $\attB$ attribute, say $\beta$, in
    $F''(\Rbeg)$, where $F''=(f_1,\ldots,f_{i-1})$. Then, let
    $f'_i=(\attA\in[a_i,z_i],\attB\gets\attB+(b_i-\beta))$. If $F'$ is defined as
    $F$ where $f_i$ is replaced with $f'_i$, then $F'$ would still yield
    $F'(\Rbeg)=(\Rend)$, but $F'$ has fewer assignment modifiers than $F$.

    Therefore, $F$ has no assignment modifiers. Thus, $F$ is also
    a best diff in $\diff(\Rbeg,\Rend,\famop{\range}{+})$.
\end{proof}

\begin{proof} [of Lemma~\ref{barrier}]
    Suppose a diff between $\Rbeg$ and $\Rend$ contains an operation
    \[f=(\attA\in\bigcup_{j=1}^{r} [a_{j},z_{j}],\attB\gets b)\]
    that matches $\attA=4k$---that is,
    there exists $j\in[r]$ such that $a_j\leq 4k\leq z_j$---for
    some $k\in\{0,\ldots,n\}$.
    Because, by construction, there are two tuples with $\attA=4k$, this operation
    changes their $\attB$ values to the same value $b$. However, these tuples
    have different $\attB$ values in $\Rend$ ($-1$ and $-2$), and assignment
    operators cannot assign different $\attB$ values to them, a contradiction.
\end{proof}

\begin{proof} [of Theorem~\ref{multeqgets}]
The problem is similar to the view synthesis problem with unions
of conjunctive queries with equality predicates, which is
$\NP$-hard~\cite{das2010synthesizing}.

When a view $V$ respective to attributes in $\calA$ is desired,
we set the $\attB$ values as follows:
in $\Rbeg$, set all $\attB$ values to distinct positive values;
in $\Rend$, set $\attB$ to the same values as in $\Rbeg$, except when the
tuple is in $V$, in which case $\attB$ is set to $0$.

The claim is that there is a view definition $V$ of cost $t$ if and only if
$\Rbeg$ and $\Rend$ has a diff under $\famop{=}{\gets}$ of cost $t$.
This is because including a conjunctive query into the view $V$ corresponds to
using the same query to set $\attB\gets 0$.

In fact, one modification to the reduction above is required to prevent
``unsetting'' the $\attB$ value
from $0$. For each tuple, make multiple, say $99$, copies preassigned
with different positive $\attB$ values.
Once the $\attB$ values of these
$99$ tuples are set to $0$ (or any value),
they cannot be changed back into distinct $\attB$ values again
using assignment modifier.

Thus, the problem is $\NP$-hard.
\end{proof}